\documentclass[%
reprint,
superscriptaddress,
pra,
]{revtex4-2}

\usepackage{graphicx, subfigure, float}
\usepackage{dcolumn}
\usepackage{bm}
\usepackage{hyperref}
\usepackage{physics}
\usepackage{amsthm}
\usepackage{bbold}
\usepackage{thm-restate}
\usepackage{upgreek}
\usepackage{mathtools} 

\newtheorem{lemma}{Lemma}

\newtheorem{defi}{Definition}
\newtheorem{prop}{Proposition}

  {\par\noindent\textbf{Remark: }\ignorespaces}%
  {\par}

\usepackage{xcolor}

\usepackage[normalem]{ulem}

\makeatletter
\newcommand{\smalltag}{%
  \def\tagform@##1{\maketag@@@{\scriptsize(\ignorespaces##1\unskip\@@italiccorr)}}%
}
\makeatother

\newcommand{\PP}{\mathbf{P}}
\newcommand{\QQ}{\mathbf{Q}}

\newcommand{\RR}{\mathbf{R}}
\newcommand{\ZZ}{\mathbf{S}}

\begin{document}

\title{No-local-broadcasting theorem for non-signalling behaviours and assemblages}

%
\author{Adrian Solymos}%
\email{solymos.adrian@wigner.hun-ren.hu}
\affiliation{HUN-REN Wigner Research Centre for Physics,  H-1525 Budapest P.O. Box 49, Hungary}
\affiliation{Eötvös Loránd University, Pázmány Péter sétány 1/A-C, 1117 Budapest, Hungary}

\author{Carlos Vieira}
\email{carloshv@unicamp.br}
\affiliation{Departamento de Matem\'{a}tica Aplicada, Instituto de Matem\'{a}tica, Estat\'{i}stica e Computa\c{c}\~{a}o Cient\'{i}fica, Universidade Estadual de Campinas, 13083-859, Campinas, S\~{a}o Paulo, Brazil}

\author{Cristhiano Duarte}
 \email{cristhianoduarte@gmail.com}
\affiliation{Institute for Quantum Studies, Schmid College of Science and Technology, Chapman University, One University Drive, Orange, CA, 92866, USA}
\affiliation{Instituto de Física, Universidade Federal da Bahia, Campus de Ondina, Rua Barão do Geremoabo, s.n., Ondina, Salvador, BA 40210-340, Brazil}%
\affiliation{Fundação Maurício Grabois, R. Rego Freitas, 192 - República, São Paulo - SP, 01220-010, Brazil}
\affiliation{Universidade Federal de Juiz de Fora, Departamento de Física, Juiz de Fora, MG, Brazil}

\author{Zoltán Zimborás}%
\email{zoltan.zimboras@helsinki.fi}
\affiliation{Department of Physics, University of Helsinki, Gustaf Hällströmin katu 2, Helsinki, 00014, Finland}
\affiliation{HUN-REN Wigner Research Centre for Physics,  H-1525 Budapest P.O. Box 49, Hungary}
\affiliation{Eötvös Loránd University, Pázmány Péter sétány 1/A-C, 1117 Budapest, Hungary}

\date{\today}

\begin{abstract}
The no-broadcasting theorem is a fundamental result in quantum information theory. It guarantees that a class of attacks on quantum protocols, based on eavesdropping and indiscriminate copying of quantum information, are impossible. Due to its fundamental importance, it is natural to ask whether it is an intrinsic quantum property or whether it also holds for a broader class of non-classical theories. To address this question, one could use the framework of correlation scenarios. Under this standpoint, Joshi, Grudka, and the Horodeckis conjectured that one cannot locally broadcast nonlocal behaviours. 
In this paper, we prove their conjecture based on the monotonicity of the relative entropy for behaviours. Additionally, following a similar reasoning, we obtain an analogous no-go theorem for steerable assemblages.
\end{abstract}

\maketitle
\section{Introduction}

Among the many no-go theorems in quantum information theory, it is fair to assume that the most famous are the no-cloning and the no-broadcasting theorems. Broadly speaking, the no-cloning theorem states that there is no universal quantum cloning machine that can always create two exact independent copies ($\rho\otimes\rho$) from unknown states $\rho$ \cite{wooters1982nocloning}. To a certain extent, the no-broadcasting result can be thought of as a generalisation of the no-cloning theorem. It says that there is no universal quantum channel such that, given an unknown input quantum state $\rho$, it outputs a bipartite quantum state whose either marginal coincides with the original state $\rho$ \cite{barnum1996nobroadcast, PhysRevLett.100.210502}. One might read this impossibility of broadcasting quantum information as a fundamental limit to informational tasks.  An even stronger formulation of the no-broadcasting theorem states that given two non-commuting quantum states $\rho_1$ and $\rho_2$, there is no quantum channel that could broadcast both of them \cite{barnum1996nobroadcast}.

No-cloning and no-broadcasting play a vital role in fields such as quantum cryptography, as they show that eavesdroppers are unable to perfectly copy unknown quantum information transmitted between two sites \cite{BB84,gisin2002quantumcryptography,pirandola2020advances}.
Additionally, these no-go theorems together with the Heisenberg uncertainty principle can be used to obtain a limit to the amount of information an agent can learn from a single copy of a quantum state \cite{nielsen2000quantum, massar1995optimal}.

Following the earlier results, a generalised no-broadcasting theorem was demonstrated for any non-classical, finite-dimensional probabilistic model satisfying a no-signalling criterion \cite{barnum2008generalized}. Therefore, it can be said that the impossibility put forward by the standard versions of the no-broadcasting theorem is not distinctive to quantum theory, as it is also present in broader non-classical theories as well \cite{jokinen2024nobroad}.
 
Variations of the no-broadcasting theorem have also been discussed in the literature: local versions of the theorem are perhaps the best example, see Figure~\ref{Fig.Broadcasting}. In these variations, the focus is on the potential broadcasting of a known bipartite quantum state using local operations. If one considers only LO (Local Operations), then nothing other than classical-classical states can be broadcast \cite{piani2008nolocalbroadcasting}. It has also been shown that this kind of broadcasting is equivalent to the general one in the case of quantum theory \cite{luo2010broadcastin,luo2010decomposition}. If one considers LOCC (Local Operations and Classical Communication), then only separable states can be broadcast \cite{horodecki2004dualentanglement,yang2005irreversibility}. This result remains unchanged when the requirement of LOCC operations is relaxed and all non-entangling operations are considered \cite{piani2009relative, Joshi2013nobroadcast}.

In \cite{Joshi2013nobroadcast}, the authors considered yet another variant of the broadcasting scenario. Their idea was to explore the local broadcasting of bipartite non-signalling behaviours with binary inputs and outputs---put another way, behaviours in the usual $(2,2,2)$ Bell scenario. They consider a class of operations that transform local behaviours into local ones (ones that admit a local hidden variable model) and prove that a nonlocal behaviour cannot be broadcast in this scenario if we are restricted to these transformations. The unfortunate drawback of Ref.~\cite{Joshi2013nobroadcast} is that their proof relies heavily on some properties of the $(2,2,2)$-scenario. Our work draws from that scholarship as we want to answer the question they raise: `is there a no-local-broadcasting theorem in the general scenario?', or in other words, `is there a proof for general behaviours?'.

A similar question also arises regarding the case of assemblages \cite{schrodinger_discussion_1935, Wiseman_Steering_2007,Uola_Quantum_Steering_2020}. We have seen that the no-broadcasting theorem is not exclusive to quantum theories and that alternative local versions exist for both quantum states and behaviours. In a crude approximation, assemblages can be thought of as an object halfway between bipartite quantum states and behaviours---an object that combines probabilities and quantum states. Inspired by this naive approximation, a natural line of inquiry is whether it is possible to broadcast (with the appropriate transformations) certain classes of assemblages. In this work, we not only prove a general local no-broadcasting theorem for behaviours, but also a similar result for assemblages using fundamental properties of the relative entropy of nonlocality and the relative entropy of steering.

Beyond its foundational interest, nonlocality is by now widely recognised as an operational resource in a variety of information-theoretic tasks, including device-independent quantum key distribution \cite{ekert_quantum_1991, zapatero_advances_2023}, randomness generation and expansion \cite{pironio_random_2010, liu_device-independent_2021}, as well as advantages in communication complexity problems \cite{Brukner:2004PRL}. From this perspective, it is natural to ask whether nonlocal behaviours could be locally broadcast, as such a possibility would allow one to locally amplify or distribute nonlocal correlations without additional nonlocal resources, potentially enhancing their usefulness in these applications. Establishing fundamental limitations on such broadcasting processes is therefore essential in order to understand the structure of nonlocality as a resource. Our results show that no such local amplification is possible, even when one considers general non-signalling behaviours beyond the quantum set. More broadly, our results fit within a wider effort to characterise nonclassical correlations through operational principles \cite{Brassard2006, kofler_classical_2007, Navascus2009, Pawowski2009, Fritz2013, Cabello2013a, nogueira_unexpected_2025} and no-go theorems \cite{wooters1982nocloning, barnum1996nobroadcast, jokinen2024nobroad, kumar_pati_impossibility_2000, braunstein_quantum_2007}.

\begin{figure}[ht]
	    \includegraphics[scale=1]{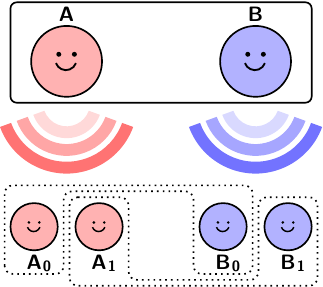}
		\caption{Pictorial representation of a local broadcasting scenario: Alice (A) and Bob (B) share some sort of known correlation, which is represented by them being enclosed in a box. They send some information to pairs of Alices ($\mathrm{A_0,A_1}$) and Bobs ($\mathrm{B_0,B_1}$) in some localised manner. The broadcast is successful if the lower Alice-Bob pairs with the same index share the original correlation illustrated by the dotted boxes (colours online).  \label{Fig.Broadcasting}}
\end{figure}  

\section{No-local-broadcasting of behaviours}
\label{sec:preliminaries}

In this section, we will state and prove a theorem about the impossibility of locally broadcasting known nonlocal behaviours (Thm.~\ref{Thm.ImpossibilityBroadcastBoxes}). We start by introducing the concepts of correlation scenarios, local behaviours, and the like. Then, we introduce a set of physically motivated `local' transformations that will become the core assumption of the no-go theorem. We conclude this section by showing that the possibility of broadcasting a nonlocal behaviour would imply a contradiction.

\subsection{Correlation scenarios}
\label{sec:preliminaries-correlation-scenarios}

Correlation scenarios are usually formulated in a device-independent framework. Think of one as a collection of $N$ boxes. Each box comes with $m$ number of inputs and $o$ number of outputs. Inputs are usually thought of as buttons that one can interact with. Pressing buttons generates an answer, which is the output for that given input on a particular box. In the device-independent framework, it is customary to imagine that each of these boxes belongs to an agent, that they are causally separated, and that agents will synchronise their actions, which means that they will push just one button per round, and a new round only starts after all agents have interacted with their boxes and written down the outcome of this interaction. Figure~\ref{Fig.DeviceIndApproach} depicts a paradigmatic scenario in which a round has taken place and every agent has pushed a button on their box obtaining an output.


%
\begin{figure}[ht]
	    \includegraphics[scale=0.5]{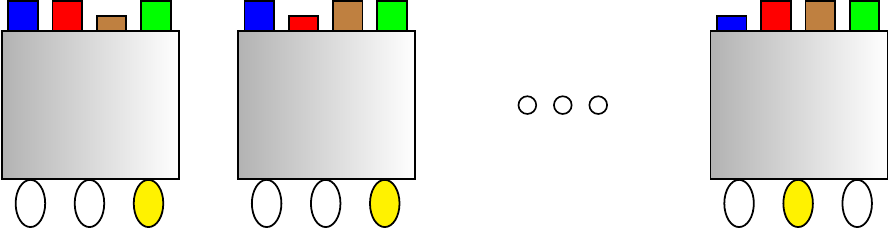}
		\caption {Pictorial representation of a (uniform) correlation scenario. Measurements are represented by the coloured buttons on the top of each box. A single outcome is returned after a measurement is performed, and schematically, one light bulb goes off (colours online).  \label{Fig.DeviceIndApproach}}
\end{figure}   

Note that in the most general scenario each box could have some arbitrary number of inputs and each of these inputs could have some arbitrary number of outputs. While these general scenarios are very much tractable, in practice one usually only deals with behaviours with a uniform number of inputs ($m$) and outputs ($o$) on each box, as it leads to less clutter in mathematical notation. Furthermore, for a behaviour's any given box, the inputs (outputs) are thought of as numbers from the set of positive integers $[m]\coloneqq \{1,2,\ldots, m\}$ ($[o]\coloneqq \{1,2,\ldots, o\}$). This ensures that every behaviour with the same number of parties, inputs, and outputs is comparable, and different input/output `alphabets' (labels) do not matter. Thus, from now on, we shall focus on $(N,m,o)-$correlation scenarios.

The framework is formulated in a manner that the inner physical mechanisms of each box are hidden away from their agents. Consequently, as they do not have access to the physical details producing a particular outcome, given the pushing of a certain button, the only reasonable description for this $(N,m,o)-$correlation scenario is via the aggregated joint statistics called the behaviour:
	\begin{equation}
		\small \PP=\{P(a,b,...,c | x,y,...,z)\}_{a,b,...,c,x,y,...,z} \in \mathbb{R}^{(om)^N}.
		\label{Eq.DefbehaviourGeneral2}
	\end{equation}


Each element $P(a,b,...,c | x,y,...,z)\in[0,1]$ means the joint probability of obtaining the outcome $a$ out of the first box, outcome $b$ out of the second box, ..., and outcome $c$ out of the $N$-th box, when the buttons $x,y,\ldots,z$ have been pressed on the respective boxes.

Note also that often the words `behaviour', `box', and `boxes' are used interchangeably in the literature. In our illustrative introduction, a behaviour is the same as the aggregated joint statistics that one can draw from boxes (plural). However, sometimes one imagines a single black box (cf.~PR-box \cite{BrunnerEtAl14}) with inputs from $N$ parties and outputs for $N$ parties. Thus, sometimes a behaviour is synonymous with a (singular) box. In this paper, we will use the term 'behaviour' to describe the joint correlation statistics of the individual boxes of the different parties.

	%

A further simplification is to only consider bipartite scenarios, that is, when $N=2$. In this case, we refer to the first box as belonging to Alice and the second as belonging to Bob, or $A$ and $B$ for short.

In the following, we present the most studied subsets of general correlation scenarios.

\subsection{Local behaviours}
\label{sec:preliminaries-local-boxes}

Local behaviours are considered to describe correlations having a classical explanation.

\begin{defi} \label{Def.LocalBox}
In a given $(N,m,o)-$scenario, we say that a behaviour $\PP=\{P(a,b,...,c | x,y,...,z)\}_{a,b,...,c,x,y,...,z}$ is \emph{local} whenever there exists a probability distribution $\{r(\lambda)\}_{\lambda}$ over an exogenous variable $\lambda$ and there exist conditional probability distributions $\{P_{A}(a|x,\lambda)\}_{a,x,\lambda}$, $\{P_{B}(b|y,\lambda)\}_{b,y,\lambda}$, ..., $\{P_{C}(c|z,\lambda)\}_{c,z,\lambda}$ such that
    \begin{equation}\label{Eq.DefLocal}
    \begin{split}
		&P(a,b,\ldots,c|x,y,\ldots,z)=\\
        &=\sum_{\lambda}r(\lambda)P_{A}(a|x,\lambda)  P_{B}(b|y,\lambda)  \ldots  P_{C}(c|z,\lambda)\text{,}
    \end{split}
    \end{equation}

for all inputs $x,y,\ldots, z$ and outputs $a,b,\ldots, c$.
\end{defi}

In sum, locality for a behaviour means that there is a hidden variable, to which we do not have access, that fully explains the correlation across the boxes, or more succinctly, that there is a \emph{classical} explanation for such correlations.


\subsection{Quantum behaviours}
\label{sec:preliminaries-Quantum-boxes}
Simply put, quantum behaviours describe correlations arising from the local measurement of quantum systems using positive operator-valued measures (POVMs). It can be more formally stated as follows.

\begin{defi}
In a given $(N,m,o)-$scenario, we say that a behaviour $\PP=\{P(a,b,...,c|x,y,...,z)\}_{a,b,...,c,x,y,...,z}$ is \emph{quantum}, whenever there exists a density operator  $\rho_{AB\ldots C}$  acting on $\mathcal{H}_A\otimes\mathcal{H}_B\otimes \cdots \otimes \mathcal{H}_C$, and for each input $x,y,\ldots,z$ there exists a POVM $\{ \Pi_{a}^{x} \}_a$, $\{ \Pi_{b}^{y} \}_b$, ..., $\{ \Pi_{c}^{z} \}_c$ respectively acting on the Hilbert spaces $\mathcal{H}_A$, $\mathcal{H}_B$, ..., $\mathcal{H}_C$ such that
\begin{equation}
    \small P(a,b,...,c|x,y,...,z) = \trace(\Pi_{a}^{x}\otimes \Pi_{b}^{y} \otimes \cdots \otimes \Pi_{c}^{z} \rho_{AB\ldots C})\text{,}
\end{equation}
for all inputs $x,y,\ldots, z$ and outputs $a,b,\ldots, c$.
\label{Def.QuantumCorrelations}
\end{defi}

Every local behaviour can be thought of as a quantum behaviour; however, the converse is untrue in general \cite{Bell64,CHSH69, Freedman1972experimental, aspect1982experimental}. Furthermore, there are more than just classical and quantum correlations. Behaviours that can be obtained by post-quantum theories also exist, and although they have never been realised in the laboratory, their study can still enlighten the nature of physical theories and what they all have in common. A notable example is provided by non-signalling correlations, which we will discuss next.

\subsection{Non-signalling behaviours}
\label{sec:preliminaries-NS-boxes}

Non-signalling behaviours are considered the most general ones, where the inputs of any one box do not affect the outputs of the other boxes. Every local and quantum behaviour is non-signalling; however, the reverse is untrue in general \cite{cirelson1980quantum, popescu1994quantum, aspect1982experimental, giustina2015significant}.

For the sake of simplicity, we introduce the notion of non-signalling only for two boxes. All of the reasoning here can be easily extended to larger scenarios with $N >2$ boxes. We also refer to Ref.~\cite {DDO18}, where non-signalling is motivated and rigorously defined for correlation scenarios involving more than two boxes. 

\begin{defi}
In a given $(2,m,o)-$scenario, we say that a behaviour $\PP=\{P(a,b|x,y)\}_{a,b,x,y}$ is \emph{non-signalling} whenever it satisfies for all $a,x,y,y'$
\begin{subequations}
\begin{equation}
    \sum_{b}P(a,b|x,y) \eqqcolon P(a|x) \coloneqq \sum_{b}P(a,b|x,y^{\prime}),
\end{equation}
and for all $b,y,x,x'$
\begin{equation}
    \sum_{a}P(a,b|x,y) \eqqcolon P(b|y) \coloneqq \sum_{a}P(a,b|x^{\prime},y).
\end{equation}
\end{subequations}
\label{Def.NonSignallingCorr}
\end{defi}
Simply put, one may want to consider non-signalling behaviours as those behaviours for which all marginal probabilities can be defined independently of the choices of the other parties. For the $(2,m,o)$ case, it is equivalent to saying that for the marginal of Alice, $P(a|x,y)=P(a|x)=P(a|x,y')$ for all choices of measurements and outcomes, and similarly for the marginal of Bob.

\subsection{The Kullback-Leibler divergence of behaviours}
We prove our main result for behaviours using information measures. In general terms, our argument is as follows: on the one hand, we will argue that these measures are monotones for a set of free transformations. On the other hand, we will demonstrate that broadcasting-like transformations from the aforementioned set of free transformations violate this monotonicity, consequently leading to a contradiction and implying that they cannot exist. To do so, we first need to decide upon an information measure. In this paper, we work with well-adapted extensions of the Kullback-Leibler divergence.  

The \textit{Kullback-Leibler divergence} (also called \textit{relative entropy}, or KL-divergence for short) is an information-based measure of disparity among probability distributions \cite{kullback_information_1951, kullback_information_1997}.

\begin{defi}
Let $p=\{p(a)\}_a$ and $q=\{q(a)\}_a$ be two probability distributions defined over the same finite set. The Kullback-Leibler divergence of $p$ from $q$ is the following:
\begin{equation}\label{eq:KL-divergence}
S(p||q)\coloneqq \sum_{a} p(a)\log\left(\frac{p(a)}{q(a)}\right),
\end{equation}
where we use the convention $0\log(\frac{0}{q(a)})\coloneqq 0$ for all $q(a)$, and $p(a)\log(\frac{p(a)}{0})\coloneqq\infty$ for $p(a)\neq 0$.
\end{defi}

It is well known that the KL-divergence is: $(i)$ non-negative; $(ii)$ zero if, and only if, the distributions match exactly; $(iii)$ not symmetric in $p$ and $q$; and $(iv)$ can potentially be equal to infinity. The KL-divergence is, in other words, a type of statistical distance: a measure of how one probability distribution $p$ is different from a second reference probability distribution $q$. In particular, the KL-divergence measures the information lost when $q$ is used to approximate $p$ \cite{burnham_model_2002}. It plays a central role in the theory of statistical inference \cite{cover1999elements}.

Given its extreme importance, it is natural to study extensions of the KL-divergence to other objects. For correlation scenarios, a possible generalisation was introduced in Refs.~\cite{PhysRevLett.95.210402, gallego2017nonlocality}. The main idea is to associate a single probability distribution with a behaviour, and once this is done, we use the KL-divergence defined in eq.~\eqref{eq:KL-divergence} to compare two behaviours. More concretely, in the case of a bipartite correlation scenario, if we consider a probability distribution $\pi=\{\pi(x,y)\}_{x,y}$ describing the probability of Alice's and Bob's choices for inputs, then $P(a,b,x,y) \coloneqq \pi(x,y)P(a,b|x,y)$ is the probability of them choosing the inputs $x,y$ and getting the outputs $a,b$. Therefore, we can compare two behaviours $\PP$, $\QQ$ through the KL-divergence between the probability distributions $\pi\PP\coloneqq\{\pi(x,y)P(a,b|x,y)\}_{a,b,x,y}$ and $\pi\QQ\coloneqq\{\pi(x,y)Q(a,b|x,y)\}_{a,b,x,y}$. More generally, if we want to measure how different two behaviours are, it is natural to grant Alice and Bob the freedom to choose their inputs according to any probability distribution available to them and then single out the one that optimises this distance. This process can be used to define the \textit{KL-divergence for behaviours}:

\begin{defi}
Let $\PP$ and $\QQ$ be two behaviours from the same $(2,m,o)$-scenario, and let $\pi=\{\pi(x,y)\}_{x,y}$ be a generic joint probability distribution on the inputs. The KL-divergence of behaviour $\PP$ from $\QQ$ is the following:
\begin{equation}
      S_{\mathrm{b}}(\PP||\QQ)
      \coloneqq \sup_{\pi} S(\pi\PP||\pi\QQ).
      \label{Eq.DefKLDivergenceBox}
\end{equation}
\end{defi}

Naturally, the above definition is easy to generalise to $(N,m,o)$-scenarios.

Remarkably, the expression in eq.~\eqref{Eq.DefKLDivergenceBox} can be simplified. As shown in Ref.~\cite{gallego2017nonlocality}, the optimal probability $\pi^{\ast}$ is such that the KL-divergence between the two behaviours can be expressed as:
\begin{equation}
     S_{\mathrm{b}}(\PP||\QQ)
     = \max_{x,y} S(\PP(. \, ,.|x,y)||\QQ(. \, ,.|x,y))\text{,}
\end{equation}
where $\PP(.\, ,.|x,y)\coloneqq \{P(a,b|x,y)\}_{a,b}$ denotes the probability distribution of the outcomes given the inputs $x$ and $y$, and similarly for the reference $\QQ$.

\subsection{Broadcasting of a behaviour}
We start by recalling the standard definition of broadcasting in the case of bipartite quantum states \cite{barnum1996nobroadcast}. Let $\mathcal{D}(\mathcal{H})$ denote the set of quantum states (density operators) acting on a Hilbert space $\mathcal{H}$. Given a bipartite quantum state $\rho_{AB} \in \mathcal{D}(\mathcal{H}_A\otimes\mathcal{H}_B)$, we say that the quantum state $\rho'\in \mathcal{D}(\mathcal{H}_{A_0} \otimes \mathcal{H}_{B_0} \otimes \mathcal{H}_{A_1}\otimes \mathcal{H}_{B_1})$ is a broadcast version of $\rho$ when the reduced states of $\rho'$ to the appropriate marginals ($A_0,B_0$ or $A_1,B_1$) are equal to $\rho$. Put another way, $\trace_{A_0,B_0}(\rho') = \rho = \trace_{A_1B_1}(\rho')$. It is worth emphasising that \textit{broadcasting} is a generalisation of \textit{cloning}, where we now allow for the possibility of correlations between the two parts. 

An analogous definition for behaviours can be crafted in the following way: consider a behaviour $\PP$ in the $(2,m,o)$-scenario. Given a behaviour $\PP'$ from the $(4,m,o)$-scenario describing the aggregated joint statistics between four agents labelled $A_0,A_1,B_0,B_1$, we say that it is realising a broadcasting of behaviour $\PP$ if the marginal statistics for the pair of agents $A_0, B_0$ is equal to the marginal statistics for the pair $A_1,B_1$, which must be, in turn, equal to the behaviour $\PP$. We summarise this discussion in the definition below. 

\begin{defi}[Broadcasting of a behaviour]
We say that a behaviour $\PP'$ from the $(4,m,o)$-scenario with agents $A_0,A_1,B_0,B_1$ is a broadcast version of the behaviour $\PP$ from the $(2,m,o)$-scenario with agents $A,B$ if:
\begin{subequations}\label{BroadcastingBoxes}
\begin{equation}
    \sum_{a_1,b_1}\! P'(a_0,\! a_1,\! b_0,\! b_1|x_0,\! x_1,\! y_0,\! y_1)\!=\! P(a_0,\! b_0|x_0,\! y_0),
\end{equation}
and
\begin{equation}
    \sum_{a_0,b_0}\! P'(a_0,\! a_1,\! b_0,\! b_1|x_0,\! x_1,\! y_0,\! y_1)\! =\! P(a_1,\! b_1|x_1,\! y_1),
\end{equation}
\end{subequations}
where the above equations must be true for every choice of inputs and outputs.
\label{Def.BroadcastForBoxes}
\end{defi}

Naturally, the above definition is easy to generalise to $(N,m,o)$-scenarios.

Recalling the usual notion of non-signalling behaviours, we should notice that eqs.~\eqref{BroadcastingBoxes} imply that $\PP'$ is non-signalling with respect to the partition $A_0B_0|A_1B_1$. 

Note that, trivially, a broadcast version can always be found for any behaviour. However, this paper explores whether a local-type transformation exists that can broadcast a known nonlocal behaviour.

In order not to overload the notation, we will usually denote the pair $a_0, a_1$ by $\bm{a}$ and similarly for $\bm{b}, \bm{x}$ and $\bm{y}$. Nonetheless, we will alternate between these two notations when it is convenient.

\subsection{Transforming behaviours into behaviours}

\subsubsection{LOSR Transformations}
To prove the impossibility of locally broadcasting behaviours, we must specify which set of transformations of behaviours we will consider and, more importantly, why.

The primary requirement for the transformations is that they ought to be local transformations in some sense. One should draw a parallel with LOCC transformations, which preserve the set of separable states and are, from the standpoint of quantum theory, a type of local operation.

In this sense, we will adapt LOSR (Local Operations and Shared Randomness) transformations to our needs \cite{gallego2017nonlocality,wolfe_quantifying_2020}. If the input behaviour is local these transformations do not introduce nonlocality to the output behaviour. Furthermore, they have the added benefit of being local transformations from an operational point of view, as attested by the following characterisation \cite{gallego2017nonlocality}:

\begin{defi}\label{def:LOSR}
    {\color{orange}} A map $\mathcal{M}$ from the $(2,m,o)$-scenario to the $(2,m',o')$-scenario is an \emph{elementary Local Operations and Shared Randomness} transformation if there exist behaviours $I(x,y|\chi,\psi)$, local in the $(\chi\to x )| (\psi \to y)$ split, and $O(\alpha, \beta|\chi, x, a, \psi,y,b)$, local in the $(\chi,x,a\to \alpha )| (\psi,y,b\to \beta)$ split, such that for an input behaviour $\PP=\{P(a,b|x,y)\}_{a,b,x,y}$ the output is:
    \begin{align}
        &P'(\alpha,\beta|\chi, \psi) \coloneqq \mathcal{M}(\PP)(\alpha,\beta|\chi, \psi)=\\
        &= \sum_{x,y,a,b}I(x,y|\chi,\psi)P(a,b|x,y)O(\alpha, \beta|\chi, x, a, \psi,y,b).\nonumber
    \end{align}
\end{defi}

An illustration of an elementary LOSR transformation is depicted in Fig.~\ref{fig:test1}.

As pointed out by Wolfe et al.~\cite{wolfe_quantifying_2020}, LOSR transformations should be considered as those that can be written as the convex sum of elementary LOSR transformations. In this paper, everything is also true for convex combinations of elementary LOSR transformations, given that they preserve the same locality properties as elementary LOSR transformations.

Note that LOSR transformations are usually considered as transformations of bipartite behaviours into bipartite behaviours. To discuss broadcasting, one needs to be familiar with bipartite to quadripartite transformations. One can imagine these as a standard one by bundling all inputs and outputs into a single input and output for each party. E.g.~by treating the inputs of the Alices $(x_0,x_1)$ as one input ($\bm{x}$ or $\chi$ in the definition). The drawback of \textbf{Definition \ref{def:LOSR}}, in the context of broadcasting, is that extra care needs to be taken when defining locality, see \textbf{Section \ref{subsection:NL in broadcasting scenario}}.

  \begin{center}
    \begin{figure}[h!]
    \includegraphics[width=.9\linewidth]{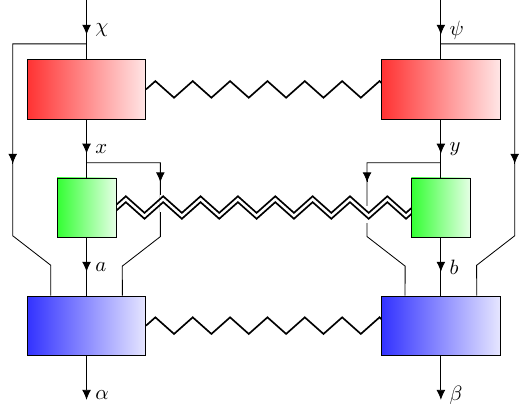}
    \caption{Illustration of an elementary LOSR transformation of a behaviour. The input behaviour $\PP$ is in green in the middle, with two wavy lines between the two parts representing possibly non-classical correlations; its inputs are $x$ and $y$, while the outputs are $a$ and $b$. The red boxes above represent the preprocessing behaviour $I$, while the blue ones below represent the post-processing behaviour $O$. They have single wavy lines between them to represent classical correlations. The final behaviour $\PP'$ has inputs $\chi,\psi$ and outputs $\alpha,\beta$ (colours online).}
  \label{fig:test1}
  \end{figure}
  \end{center}

An important property of LOSR transformations is that they transform local behaviours into local behaviours, i.e.~in our case, if $\PP$ is local, then $\PP' \coloneqq  \mathcal{M}(\PP)$ will also be local in the $A_0A_1 | B_0B_1$ partition \cite{wolfe_quantifying_2020}. Furthermore, any local behaviour can be prepared using a LOSR transformation, meaning that local behaviours can automatically be broadcast by LOSR transformations.

It should be highlighted that the set of transformations that preserves the set of local behaviours is a larger set than that of LOSR transformations. The standard examples are WPICC (Wirings and Prior-to-Input Classical
Communication) maps. These are transformations that preserve locality, but their implementation requires communication between the parties \cite{PhysRevLett.95.210402, gallego2017nonlocality}. Despite being a subset of the transformations that preserve locality, as discussed in Ref.~\cite{wolfe_quantifying_2020}, the set of LOSR transformations is the most natural set of `free' transformations in the context of nonlocality. This is due to the fact that its implementation only utilises resources that preserve the initial causal structure of the scenario, namely, classical correlations and local operations.

\subsubsection{LOSR Transformations and KL-Divergence}

An interesting property of the KL-divergence is that it is contractible over the set of LOSR transformations \cite{gallego2017nonlocality}, i.e.~if $\PP, \QQ$ are two non-signalling behaviours from the same scenario and $\mathcal{M}$ is a LOSR transformation, then 
\begin{equation}\label{eq:ContracitivityDivergence}
    S_{\mathrm{b}}\left(\mathcal{M}(\PP)||\mathcal{M}(\QQ)\right) \le S_{\mathrm{b}}(\PP||\QQ). 
\end{equation}

\subsection{Locality in the broadcasting scenario}\label{subsection:NL in broadcasting scenario}

In the previous section, we have highlighted that LOSR transformations preserve locality from the bipartite $A|B$ partition into the quadripartite $A_0A_1|B_0B_1$ partition. However, note that \textbf{Definition \ref{def:LOSR}} does not have any requirements on what locality properties the possible marginals should have given a local input. To this end, we introduce the set of \textit{Locally Realistic Non-Signalling behaviours} as defined by Ref.~\cite{Joshi2013nobroadcast} as a local set for the multipartite case. This set consists of behaviours that are local in the partition $A_0A_1|B_0B_1$, and, in addition, the local decomposition has non-signalling components between $A_0$ and $A_1$ and also between $B_0$ and $B_1$, leading to all marginals being well-defined. 

More formally,

\begin{defi}
    A behaviour $\QQ=\{Q(\bm{a},\bm{b}|\bm{x},\bm{y})\}_{\bm{a},\bm{b},\bm{x},\bm{y}}$ from the $(4,m,o)$-scenario is \emph{locally realistic non-signalling} if there exists a probability distribution $\{r(\lambda)\}_{\lambda}$ over an exogenous variable $\lambda$ and there exist behaviours $\{Q_{A}(a_0,a_1|x_0,x_1,\lambda)\}_{a_0,a_1,x_0,x_1,\lambda}$ and $\{Q_{B}(b_0,b_1|y_0,y_1,\lambda)\}_{b_0,b_1,y_0,y_1,\lambda}$ that are non-signalling for every fixed $\lambda$ in the $A_0|A_1$ and $B_0|B_1$ splits, respectively, such that every element of $\QQ$ can be decomposed as:
    \begin{equation}
    \begin{split}
       &Q(a_0, a_1, b_0, b_1|x_0, x_1, y_0, y_1) =\\
       &= \sum_{\lambda} r(\lambda)Q_{A}(a_0, a_1|x_0,x_1,\lambda)Q_{B}(b_0, b_1|y_0,y_1,\lambda)\text{.}
    \end{split}
    \end{equation}
    The set of all behaviours with such properties will be denoted by $LR_{ns}$.
\end{defi}

Note that this means that a quadripartite locally realistic non-signalling behaviour is non-signalling in every possible partition of its boxes. Furthermore, such behaviour is local in any partition of its boxes, with Alices and Bobs on different sides of the partition. This set will be considered local in the broadcasting scenario.

Having introduced the set that we can consider local in the broadcasting scenario, we can ask whether or not a behaviour belongs to this set. In the negative case, when a behaviour does not belong to this set, it is fair to still be interested in how distant it is from the elements of this set. To answer, we can use the KL-divergence defined for correlation scenarios, which, as discussed, compares how different two behaviours are. Using the KL-divergence for behaviours, we introduce the relative entropy of nonlocality in the following way \cite{gallego2017nonlocality}:

\begin{defi}
    The \emph{relative entropy of nonlocality} of a behaviour $\PP$ from the $(4,m,o)$-scenario is given by:
    \begin{equation}
        E_{\mathrm{LR}}(\PP) \coloneqq \inf_{\QQ \in \mathrm{LR_{ns}}} S_{\mathrm{b}}(\PP||\QQ),
        \label{Eq.DefRelativeEntropyOfNL}
    \end{equation}
    where the index LR references the $LR_{ns}$ set.
\end{defi}

Note that the KL-divergence for correlation scenarios is not a proper distance measure, as it is not symmetric. Nevertheless, it is fair to say that the relative entropy of nonlocality defined by eq.~\eqref{Eq.DefRelativeEntropyOfNL} measures how distinct $\PP$ is from the probability distributions that admit a locally realistic non-signalling description. 

Note that, with a slight abuse of notation, we can extend the relative entropy of nonlocality to behaviours from the $(2,m,o)$-scenario as well by treating the local realistic non-signalling set ($\mathrm{LR_{ns}}$) as the local set (L). This respects the properties of the above defined entropy since the marginals of a $\mathrm{LR}_{ns}$ behaviour are local.

Given this definition of the local set in the broadcasting scenario, an additional requirement for the LOSR transformations to be admissible as local transformations is for them to preserve this definition of locality:

\begin{defi}
A transformation of behaviours is said to be a \textit{$LR_{ns}$-LOSR transformation} when $(i)$ it belongs to the set of LOSR transformations, and $(ii)$ it transforms local behaviours into behaviours in $LR_{ns}$.
\end{defi}

Thus, $\mathrm{LR_{ns}}$-LOSR transformations will be the relevant set of transformations for local broadcasting. Note that no restrictions have been imposed for the marginals of the output for general nonlocal inputs. The only tacitly implied restriction is that for the known behaviour that one wants to broadcast (whether local or nonlocal), the marginals on $A_0B_0$ and on $A_1B_1$ must be well-defined for us to consider the transformation a broadcasting transformation. For all other input behaviours, the well-defined (non-signalling) marginals might not exist.

\subsection{No-local-broadcasting theorem for nonlocal non-signalling behaviours}
\label{sec:preliminaries-ea-pam-scenarios}
This section deals with our main no-go theorem for behaviours. We show that it is impossible to locally broadcast a nonlocal non-signalling behaviour using only $\mathrm{LR_{ns}}$-LOSR transformations. Our argument is by contradiction. In fact, we will establish that if it were possible to broadcast a nonlocal behaviour $\PP$ by a $\mathrm{LR_{ns}}$-LOSR transformation $\mathcal{M}$, then we would have $E_{\mathrm{LR}}(\mathcal{M}(\PP)) < E_{\mathrm{LR}}(\mathcal{M}(\PP))$, which is a contradiction.

The proof will be divided into two parts, which we will state as propositions. First, we will show that the relative entropy of nonlocality is contractive over the set of $\mathrm{LR_{ns}}$-LOSR transformations.
\begin{restatable}{prop}{propEntropyContractivity}\label{prop:EntropyContractivity}
    If $\mathcal{M}$ is a $LR_{ns}$-LOSR transformation and $\PP$ is a non-signalling behaviour, then
    \begin{equation}
    E_{\mathrm{LR}}(\mathcal{M}(\PP)) \le  E_{\mathrm{LR}}(\PP)\text{.}
    \end{equation}
\end{restatable}
The proof of the above proposition is mainly based on two special properties. One is the fact that $\mathrm{LR_{ns}}$-LOSR transformations preserve the set of $\mathrm{LR_{ns}}$ behaviours. The other is the contractivity of the KL-divergence over LOSR transformations (see eq.~\eqref{eq:ContracitivityDivergence}). The details of the proof are given in \textbf{Appendix \ref{sec:appendix-proof-prop:EntropyContractivity}}.

The second proposition tells us that the amount of nonlocality that a broadcast version from a nonlocal behaviour has is greater than the initial behaviour.
\begin{restatable}{prop}{propEntropyDilation}\label{prop:EntropyDilatation}
    If $\PP'$ is a broadcast version of a nonlocal behaviour $\PP$, then 
    \begin{equation}\label{eq:major}
        E_{\mathrm{LR}}(\PP') > E_{\mathrm{LR}}(\PP).
    \end{equation}
\end{restatable}
The proof of this proposition is more elaborate, and the details are given in \textbf{Appendix \ref{sec:appendix-proof-prop:EntropyDilatation}}. To summarise, we adapt the well-known \textit{chain rule} of probability distributions for the case of behaviours \cite{kullback_information_1997,cover1999elements}. The chain rule is an identity between the KL-divergence of a joint probability distribution and the KL-divergence of one of the marginals of this distribution plus the mean of the KL-divergence of the conditional distributions (see eq.~\eqref{eq:ChainRuleBox}). From this identity, we can recover the relative entropy of the marginal distribution as the first term of the identity. Furthermore, it was possible to prove that if $\PP$ is nonlocal, then the second term is strictly positive. Concluding therefore that $E_{\mathrm{LR}}(\PP') > E_{\mathrm{LR}}(\PP)$.

We can now establish one of the main results of this work.
\begin{restatable}{thm}{ThrNoBroadcastingBehaviours}[No-local-broadcasting of behaviours]\label{Thm.ImpossibilityBroadcastBoxes}
It is impossible to locally broadcast any known bipartite nonlocal non-signalling behaviour using LR$_{ns}$-LOSR transformations.
\end{restatable}

\begin{proof}
    Let us suppose, by contradiction, that there is a $\mathrm{LR_{ns}}$-LOSR transformation $\mathcal{M}$ that broadcasts a known nonlocal non-signalling behaviour $\PP$. Then, by \textbf{Proposition \ref{prop:EntropyContractivity}}, we have:
    \begin{equation}\label{eq:ThmEntropyContractivity}
        E_{\mathrm{LR}}\left(\mathcal{M}(\PP)\right) \le  E_{\mathrm{LR}}(\PP)\text{.}
    \end{equation}
    On the other hand, as $\mathcal{M}(\PP)$ is a broadcast version of $\PP$, by \textbf{Proposition \ref{prop:EntropyDilatation}}:
    \begin{equation}\label{eq:ThmEntropyDilation}
        E_{\mathrm{LR}}(\PP) <  E_{\mathrm{LR}}\left(\mathcal{M}(\PP)\right).
    \end{equation}
    Therefore, by Equations \eqref{eq:ThmEntropyContractivity} and \eqref{eq:ThmEntropyDilation}:
    \begin{equation}
        E_{\mathrm{LR}}(\mathcal{M}(\PP)) < E_{\mathrm{LR}}(\mathcal{M}(\PP)),
    \end{equation}
    which is a contradiction. Thus, it is impossible to broadcast any known nonlocal non-signalling behaviour by $\mathrm{LR_{ns}}$-LOSR transformations. 
\end{proof}

We emphasise that Theorem~\ref{Thm.ImpossibilityBroadcastBoxes} holds at the level of general non-signalling behaviours and is therefore not restricted to those admitting a quantum realisation. In particular, the result applies to genuinely post-quantum behaviours, such as PR boxes \cite{popescu1994quantum}.

\section{No-local-broadcasting of steerable assemblages}
\label{sec:results}
In this section, we show that steerable assemblages \emph{cannot} be broadcast by local operations. To prove this claim, we adopt a strategy essentially analogous to the one we implemented for behaviours. However, several adaptations are needed. In a bipartite steering scenario, the situation is described by a collection of ensembles of quantum states on Bob’s side together with a conditional probability distribution (box) on Alice’s side. Thus, from the outset, there is a combination of quantum states and probability distributions. This combination is reflected in our proof, as we combine the strategies from the previous section with those from the no-local-broadcasting theorem for quantum states \cite{piani2009relative}.

We begin by recalling the fundamental concepts behind quantum steering; more extensive discussions can be found in \cite{schrodinger_discussion_1935, Wiseman_Steering_2007, Uola_Quantum_Steering_2020}. Then, based on the KL-divergence, we provide a definition of relative entropy suitable for the steering scenario. Following that, we address a set of physically motivated transformations that will become the basis of our analysis for the no-go theorem with which we conclude this section. 

\subsection{Steering scenarios}

Similar to correlation scenarios, steering scenarios are often determined by the number of agents involved in the process, how those agents are grouped, and also by the resources to which those agents have access. Again, one can imagine granting boxes to each agent; however, some fixed agents will be allowed to open their box and realise that it contains a quantum system, and thus, they can do a tomographically complete measurement on it. The remaining agents will get simple boxes with inputs (buttons) and outputs (lights), and no information about the contents. Again, the agents work in rounds: some just press a button at the beginning of a round and receive an outcome at the end, while the others perform a tomographically complete measurement and obtain partial information about the quantum state in the box for that given round. 

Figure~\ref{Fig.DeviceIndAssembl} depicts the simplest steering scenario where one party has access to a box with some inputs and some outputs, and the other party has access to an `openable' box that contains a quantum system.
\begin{figure}[ht]
	    \includegraphics[scale=0.65]{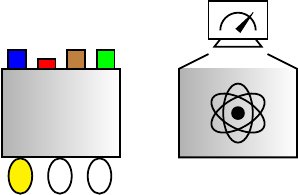}
		\caption{Pictorial representation of a bipartite steering scenario. One agent gets the usual box with coloured buttons on the top and outcomes represented by light bulbs. The other agent receives a quantum system on which they can make measurements. After pressing buttons and performing measurements, every agent receives an outcome (colours online).  \label{Fig.DeviceIndAssembl}}
\end{figure}

After many rounds of observations, the parties collect their statistics into a so-called \emph{assemblage}: a set containing probabilities for the given input-output pairs and the related quantum state of the agents who were provided a quantum state. Formally:
\begin{equation}
    \small \RR\coloneqq\{ \!P(a,\!b,\!...,\! c|x,\!y,\!...,\! z) \rho_{a,b,...,c|x,y,...,z}\}_{a,b,..., c,x,y,...,z}\text{,}
\end{equation}
where the probability distribution has $x,y,\ldots, z$ inputs and $a,b,\ldots,c$ outputs related to the number of agents who have received boxes, whereas the quantum state $\rho_{a,b,\ldots ,c|x,y,\ldots, z}$ has as many parts as agents who have received quantum states.

Similar to correlation scenarios, the central purpose of steering theory is to describe correlations between the parties involved. Adopting a realistic interpretation of quantum theory, one could say: do the choices of the parties with boxes influence (steer) the quantum states on which the other parties can make measurements?

While general steering scenarios are very much tractable, in practice one usually only deals with assemblages with a uniform number of inputs ($m$) and outputs ($o$), and simple numerical input-output alphabets on each box and uniformly $d$-level quantum systems for those parties that receive a quantum system. Thus, one can talk about $(N_1,N_2,m,o,d)$-scenarios where $N_1$ is the number of parties with boxes of $m$ inputs and $o$ outputs, and $N_2$ is the number of parties with $d$-level quantum systems. For these uniform steering scenarios, it is less of a hassle to pinpoint where the aggregated joint statistics belong to seen as a vector:
\begin{equation}
    \RR\in \underbrace{\mathcal{L}(\mathcal{H}^{\otimes N_2}) \times \cdots \times \mathcal{L}(\mathcal{H}^{\otimes N_2}}_{(m \cdot o)^{N_1}})\text{,}
\end{equation}
where $\mathcal{L}(\mathcal{H}^{\otimes N_2})$ denotes the linear operators acting on $\mathcal{H}^{\otimes N_2}$, the joint $d^{N_2}$-dimensional Hilbert space, which describes the quantum systems of the parties having access to quantum systems $\mathcal{H}$.

For concreteness, in the following we consider the well-known bipartite scenario ($N_1=N_2=1$), involving two parties (two agents), namely Alice and Bob, or $A$ and $B$ for short. An assemblage between them is described by the following set:
\begin{equation}
    \RR\coloneqq\{P(a|x) \rho_{a|x}\}_{a,x}  \in\underbrace{\mathcal{L}(\mathcal{H}_{B}) \times \cdots \times \mathcal{L}(\mathcal{H}_{B}}_{m \cdot o})\text{,}
\end{equation}
where $\mathcal{H}_B$ denotes the Hilbert space describing the quantum system Bob has access to.

In the following, we present the most studied subsets
of bipartite steering scenarios.

\subsection{Unsteerable assemblages}

The modern bipartite steering scenario emerged as an asymmetric information-theoretic task, in which Alice attempts to convince Bob that she can generate an entangled bipartite state $\rho_{AB}$ by providing him with half of it, $\rho_B$ \cite{Wiseman_Steering_2007}. Alice's choice of measurement and outcome (linked to the effect $\Pi^x_a$) on her part of the state might be correlated with the state of Bob's quantum system. More concretely, the probability of Alice getting the output $a$ when measuring $x$ is
\begin{equation}
    P(a|x) \coloneqq \trace[(\Pi_a^x \otimes \mathbb{1}_B )\rho_{AB} ]\text{,}
\end{equation}
and, if $P(a|x)\neq 0$, then the post-measurement state of Bob is
\begin{equation}
    \rho_{a|x} \coloneqq \frac{\trace_{A}[(\Pi_a^x \otimes \mathbb{1}_B )\rho_{AB}]}{P(a|x)}\text{,}
\end{equation}
otherwise, it is non-existent.

When the state Alice produces is highly entangled (e.g.~a Bell state), one interpretation is that Alice's measurement choice steers Bob's part of the quantum state (the reduced state of Bob) into a different ensemble, out of which an element can be post-selected if Alice's outcome is known. However, if the state Alice produces is separable, then she cannot have a non-classical influence on Bob's system. In this case, Bob's part of the quantum state may only appear to be influenced: any influence can be explained by an underlying local hidden state model that Alice uses to present different ensembles corresponding to her different measurement choices, utilising classical correlations. More formally:

\begin{defi}
    In a bipartite $(1,1,m,o,d)$ steering scenario, we say that an assemblage $\RR=\{P(a|x) \rho_{a|x}\}_{a,x}$ is \emph{unsteerable} whenever there exists a probability distribution $\{r(\lambda)\}_\lambda$ over an exogenous variable $\lambda$, together with a set of quantum states $\{\rho(\lambda)\}_\lambda\subset \mathcal{D}(\mathcal{H}_B)$ and a probability distribution $\{P(a|x,\lambda)\}_{a,x,\lambda}$, such that
    \begin{equation}\label{eq:LHS-model}
        P(a|x)\rho_{a|x}=\sum_\lambda r(\lambda)P(a|x,\lambda) \rho(\lambda)\text{,}
    \end{equation}
    for all inputs $x$ and outputs $a$.
\end{defi}

We should note the similarity between the definition of local behaviours, eq.~\eqref{Eq.DefLocal}, and unsteerable assemblages, eq.~\eqref{eq:LHS-model}. Both definitions follow the same idea that correlations between parts can be generated from and mediated by a classical hidden variable $\lambda$. If an assemblage is unsteerable, we say that it admits a local hidden state (LHS) model. Assemblages that do not admit a local hidden-state (LHS) model are said to be steerable \cite{Wiseman_Steering_2007}.

Note that unsteerability does not imply separability of the underlying quantum state, as there are bipartite entangled states that are unsteerable \cite{Wiseman_Steering_2007}, similarly to the fact that there exist bipartite entangled states that do not exhibit nonlocality \cite{werner1989states}.

\subsection{Non-signalling assemblages}

Non-signalling assemblages are considered the most general ones where the input of Alice does not change the reduced state on Bob's side. More formally:

\begin{defi}
     In a bipartite $(1,1,m,o,d)$ steering scenario, we say that an assemblage $\RR=\{P(a|x) \rho_{a|x}\}_{a,x}$ is \emph{non-signalling} whenever it satisfies for all inputs $x,x'$
     \begin{equation}
         \sum_a P(a|x)\rho_{a|x}\coloneqq \rho_B\eqqcolon\sum_a P(a|x')\rho_{a|x'}\text{.}
     \end{equation}
\end{defi}
This condition has a clear operational interpretation: if Alice could choose an input $x$ such that Bob’s reduced state $\sum_a P(a|x)\rho_{a|x}$ depends on $x$, then Bob could infer Alice’s choice by performing local tomography on his system, enabling signalling from Alice to Bob. The non-signalling condition therefore guarantees that Alice’s measurement choice cannot influence Bob’s local statistics.  Observe that unsteerable assemblages are trivially non-signalling.

Note that in the bipartite steering scenario, any non-signalling assemblage can be linked to an underlying bipartite quantum state that Alice and Bob make measurements on \cite{Uola_Quantum_Steering_2020}. However, signalling bipartite assemblages and general multipartite assemblages may not have an underlying quantum state that generates them. In the following, we abandon the requirement of having an underlying quantum state description and refer only to the assemblage description.

We can easily define non-signalling for multipartite assemblages as well \cite{Uola_Quantum_Steering_2020}, which we shall do for the two-Alices-scenario:
\begin{defi}[Multipartite non-signalling]\label{def:multipart_ns}
    In a multipartite $(2,1,m,o,d)$ steering scenario with two Alices and one Bob, we say that the assemblage $\RR=\{P(a_0,a_1|x_0,x_1)\rho_{a_0,a_1|x_0,x_1}\}_{a_0,a_1,x_0,x_1}$ is $A_0\to A_1B$ non-signalling if for all $x_0,x_0'$ inputs of $A_0$ we have:
    \begin{equation}
    \begin{gathered}
        \sum_{a_0}P(a_0,a_1|x_0,x_1)\rho_{a_0,a_1|x_0,x_1}\\ \eqqcolon P(a_1|x_1)\rho_{a_1|x_1}\coloneqq\\
        \sum_{a_0}P(a_0,a_1|x_0',x_1)\rho_{a_0,a_1|x_0',x_1}\text{,}
    \end{gathered}
    \end{equation}
    and similarly for the $A_1\to A_0B$ direction.
\end{defi}

In particular, for the behaviours we also have that $P(a_1|x_1)=\sum_{a_0}P(a_0,a_1|x_0,x_1)$ for any input $x_0$.

In the case of $A_0\to A_1B$ (or $A_1\to A_0B$), non-signalling means that the marginal $A_1B$ (or $A_0B$) is well defined, independently of the inputs of the other party. Indeed, without the requirement that the expression
\begin{equation}
    \sum_{a_0} P(a_0,a_1|x_0,x_1)\rho_{a_0,a_1|x_0,x_1}
\end{equation}
be independent of $x_0$, the $A_1B$ marginal would depend on the input choice of $A_0$, and thus would not correspond to a unique operational state accessible to $A_1$ and Bob. The non-signalling condition ensures that this marginal can be unambiguously identified as $\{P(a_1|x_1)\rho_{a_1|x_1}\}_{a_1,x_1}$, independently of the actions of the other party, which is essential for formulating broadcasting transformations where all output copies are required to share the same reduced assemblage. Naturally, the above definition can easily be generalised for more Alices and Bobs.

\subsection{The Kullback-Leibler divergence of assemblages}

We proved the no-local-broadcasting theorem for behaviours using information-theoretic techniques: we defined an information measure (which is a monotone for a physically motivated set of transformations), which indicated the unwarranted increase of resources if the behaviour was broadcast, leading to a contradiction. As we shall see, the strategy for the steering case follows the same lines. In this sense, we need to decide upon an appropriate monotone. For obvious reasons, we will start with the quantum version of the Kullback-Leibler divergence.

\begin{defi}
Let $\rho,\sigma\in\mathcal{D}(\mathcal{H})$ be two quantum states. The \emph{quantum Kullback-Leibler divergence} or \emph{quantum relative entropy}, of $\rho$ from $\sigma$ is the following:
\begin{equation}\label{eq:quantum-KL-divergence}
S_{\mathrm{q}}(\rho||\sigma)\!\coloneqq \! \begin{cases} \trace\left[\rho(\log(\rho) - \log(\sigma))\right]\text{,} &\text{if } \rho^0\leq\sigma^0\text{,} \\ +\infty\text{,} &\text{otherwise,} \end{cases}
\end{equation}
where $\rho^0$ and $\sigma^0$ denote the projections onto the supports of these operators, and the convention $\log(0)\coloneqq 0$ is used.
\end{defi}

The quantum KL-divergence was first studied by Umegaki \cite{umegaki_conditional_1962} as the non-commutative extension of the KL-divergence --- eq.~\eqref{eq:KL-divergence}. It also has the same statistical interpretation as its classical analogue: it tells us how different the state $\rho$ is from the state $\sigma$ \cite{hiai_proper_1991, Vedral_RelativeEntropy_2002}.

Analogously to the classical KL-divergence, the quantum KL-divergence is $(i)$ non-negative; $(ii)$  zero if and only if the quantum states are equal; $(iii)$ not symmetric in $\rho$ and $\sigma$; and $(iv)$ can potentially be equal to infinity.

Another interesting property of the quantum KL-divergence is the contractivity over quantum channels \cite{lindblad_completely_1975}. Indeed, for any completely positive trace-preserving (CPTP) map $\mathcal{E}:\mathcal{L}(\mathcal{H})\to\mathcal{L}(\mathcal{H}')$, we have that $S_{\mathrm{q}}(\mathcal{E}(\rho)||\mathcal{E}(\sigma)) \le S_{\mathrm{q}}(\rho||\sigma)$. It is precisely this property that we will use in the following paragraphs.

Before defining the KL-divergence for assemblages, we introduce an alternative representation of any assemblage as a set of classical-quantum states $\{\rho_{FB}(x)\}_{x=1}^{m} \subset \mathcal{D}(\mathcal{H}_{F} \otimes \mathcal{H}_{B})$ such that:
\begin{equation}\label{eq:inputdependentcqstates}
    \rho_{FB}(x) \coloneqq \sum_{a=1}^{o} P(a|x)  \ketbra{a}{a} \otimes \rho_{a|x},
\end{equation}
where $\mathcal{H}_F$ is an auxiliary quantum system and $\{\ket{a}\}_{a=1}^{o}$ is an orthonormal basis in $\mathcal{H}_F$.

Furthermore, if we also take into account the probability distribution of the choices of inputs of part $A$, we can even associate a unique classical-quantum state to an assemblage. Assuming that $\{\pi(x)\}_{x=1}^{m}$ is such a probability distribution, we define $\rho_{EFB} \in \mathcal{D}(\mathcal{H}_{E}\otimes \mathcal{H}_{F} \otimes \mathcal{H}_{B})$ to be the classical-quantum (CQ) state associated to an assemblage $\RR=\{P(a|x)\rho_{a|x}\}_{a,x}$ and a probability distribution $\{\pi(x)\}_{x=1}^{m}$ as follows:
\begin{equation}\label{eq:CQ-state-assemblage}
\begin{split}
    \rho_{EFB} &\coloneqq \sum_{x=1}^{m} \pi(x) \ketbra{x}{x} \otimes  \rho_{FB}(x)\\
    &=  \sum_{x=1}^{m}\sum_{a=1}^{o} \pi(x)P(a|x) \ketbra{x}{x} \otimes \ketbra{a}{a} \otimes \rho_{a|x}.
\end{split}
\end{equation}

The sets of vectors $\{\ket{x}\}_{x=1}^{m}$ and $\{\ket{a}\}_{a=1}^{o}$ are orthonormal bases of the auxiliary Hilbert spaces $\mathcal{H}_{E}, \mathcal{H}_{F}$, respectively. The states $\ket{x}$ and $\ket{a}$ are merely abstract flag states to express the inputs and outputs of part $A$, respectively, and do not describe the system inside the box on $A$'s side~\cite{PhysRevA.79.032336, PhysRevX.5.041008}.

Using this representation, one can define the \textit{KL-divergence for assemblages} \cite{PhysRevX.5.041008,PhysRevA.96.022332}.
\begin{defi}
    Let $\RR$ and $\ZZ$ be two assemblages from the same $(1,1,m,o,d)$-scenario, and let $\pi=\{\pi(x)\}_{x}$ be a generic probability distribution on the inputs of Alice. The \emph{KL-divergence of assemblage $\RR$ from $\ZZ$} is the following:
\begin{equation}
      S_{\mathrm{a}}(\RR||\ZZ) \coloneqq \sup_{\pi} S_{\mathrm{q}}(\rho_{EFB}||\sigma_{EFB}),
\end{equation}
    where $\rho_{EFB}$ and $\sigma_{EFB}$ are the classical-quantum states associated to $\RR$ and $\ZZ$, respectively, using the probability distribution $\{\pi(x)\}_{x}$.
\end{defi}

Naturally, the above definition is easy to generalise to $(N_1,N_2,m,o,d)$-scenarios.

\subsection{Broadcasting of an assemblage}

Similarly to the previously discussed definitions of broadcasting quantum states and behaviours, in the case of assemblages, the respective marginals of the broadcast version of an assemblage must be the same as those of the original assemblage. That is, for the four-partite broadcast version of the assemblage with parties $A_0,A_1,B_0,B_1$, the marginals for $A_0,B_0$ and $A_1,B_1$ have to be the same as the original assemblage.

\begin{defi}[Broadcasting of an assemblage]
An assemblage $\RR'=\{P'(a_0,a_1|x_0,x_1)\rho'_{a_0,a_1|x_0,x_1}\}_{a_0,a_1,x_0,x_1}$ from the $(2,2,m,o,d)$ steering scenario, with agents $A_0,A_1,B_0,B_1$ is called a broadcast version of the assemblage $\RR=\{P(a|x)\rho_{a|x}\}_{a,x}$ from the $(1,1,m,o,d)$ steering scenario with agents $A,B$ if:
\begin{subequations}\label{eq:BroadcastingAssemblages2}
\begin{equation}\smalltag
     \small \sum_{a_1}\! P'(a_0,a_1|x_0,x_1)\!\tr_{B_1}\!(\rho'_{a_0,a_1|x_0,x_1}\!) \! =\!\! P(a_0|x_0)\rho_{a_0|x_0}\text{,}
\end{equation}
and
\begin{equation}\smalltag
     \small \sum_{a_0}\! P'(a_0,a_1|x_0,x_1)\!\tr_{B_0}\!(\rho'_{a_0,a_1|x_0,x_1}\!) \! =\!\! P(a_1|x_1)\rho_{a_1|x_1}\text{,}
\end{equation}
\end{subequations}
where the above equations must be true for every choice of inputs and outputs.
\label{Def.BroadcastingForAssemblages}
\end{defi}

Naturally, the above definition is easy to generalise to $(N_1,N_2,m,o,d)$-scenarios.

Notice that while the broadcasting of behaviours and states is defined via marginalising or tracing out larger objects into smaller ones, our definition uses both marginalisation and partial tracing. This combination should be viewed as reflecting the fact that assemblages are situated between behaviours and states. 

We should notice that eqs.~\eqref{eq:BroadcastingAssemblages2} imply that the assemblage is non-signalling in the $A_1\to A_0B_0$ and $A_0\to A_1B_1$ directions. Recalling the notion of non-signalling for behaviours, this implies that the behaviour shared by $A_0$ and $A_1$ is non-signalling with respect to the $A_0|A_1$ partition. For ease of notation, we shall say that such a behaviour is non-signalling in the $A_0B_0|A_1B_1$ split, since these marginals can be defined independently of the inputs of the other party.

If the original assemblage $\RR$ was non-signalling, this also implies that the assemblage is non-signalling in the $A_0A_1\to B_0$ and $A_0A_1\to B_1$ directions. Note, however, that even with this requirement, signalling cannot be ruled out in the $A_0A_1\to B_0B_1$ direction, as one could have different quantum states dependent on the choices of the Alices, while their marginals are always the same (e.g.~different mixtures of maximally entangled pure states all having the maximally mixed state as marginals).

Finally, note that, trivially, a broadcast version can always be found for any assemblage. However, this paper explores whether a local-type transformation exists that can broadcast a known steerable assemblage.

\subsection{Transforming assemblages into assemblages}

Analogously to the case of correlation scenarios, we are interested in transformations that do not require communication between parts $A$ and $B$. Nevertheless, any kind of local operations and classical correlations are, in principle, allowed between them. In this way, we re-enter the LOSR paradigm, with an important difference compared to the case of behaviours: Bob, now, acts on a quantum state \cite{zjawin2023quantifyingepr}. 
\begin{defi}\label{def:LOSRa}
    A map $\mathcal{M}$ from the $(1,1,m,o,d)$ steering scenario to the $(1,1,m',o',d')$ steering scenario is a \emph{LOSR transformation between assemblages} if there exists a probability distribution $\{r(\lambda)\}_\lambda$ over an exogenous variable $\lambda$, and a pre-processing behaviour $\{I(c,x|\chi,\lambda)\}_{c,x,\chi,\lambda}$ and CPTP maps $\{\mathcal{E}_{\lambda}\}_\lambda$, along with a post-processing behaviour $\{O(\alpha|a,c)\}_{\alpha,a,c}$, such that for an input assemblage $\RR=\{P(a|x)\rho_{a|x}\}_{a,x}$ the output is:
    \begin{align}\label{eq:LOSRtrans-bipartite} 
    &P'(\alpha|\chi)\rho'_{\alpha|\chi}\coloneqq\mathcal{M}(\RR)(\alpha|\chi)=\\ \nonumber
    &=\sum_{\lambda,c}\sum_{a,x}r(\lambda)I(c,x|\chi,\lambda)P(a|x)O(\alpha|a,c)\mathcal{E}_\lambda(\rho_{a|x})\text{.}
    \end{align}
\end{defi}

See Fig.~\ref{fig:test2} for an illustration. Note that the internal variable $c$ can only be a function of $\chi$ and $\lambda$, and could be replaced by them. Note also that given the explicit convex combination structure, this is the most general version of LOSR transformations for assemblages \cite{zjawin2023quantifyingepr}.

\begin{center}
  \begin{figure}[ht]
  \includegraphics[width=\linewidth]{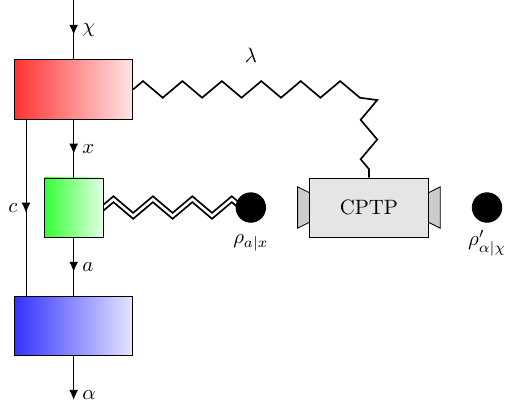}
    \caption{Illustration of an LOSR transformation of an assemblage. The green box and the black dot in the middle represent the original assemblage with input $x$ and output $a$ and related quantum state $\rho_{a|x}$, with the two wavy lines between the two parts representing possibly stronger than classical correlations. The red box represents the preprocessing probability distribution $I$, with a wavy line ($\lambda$) representing the classical correlation linking it with the CPTP map $\mathcal{E}_\lambda$. The blue box represents the post-processing probability distribution $O$. The final assemblage has input $\chi$, outputs $\alpha$, and the related quantum state $\rho'_{\alpha|\chi}$ (colours online).}
  \label{fig:test2}
  \end{figure}
\end{center}

Similarly to the case of behaviours, one wishes to discuss bipartite to quadripartite steering maps to enable broadcasting. To do this, one can think of treating the input $\chi$ and the output $\alpha$ as multiple inputs, e.g.~$(x_0,x_1)$ and $(a_0,a_1)$. Again, some extra care is needed when discussing steerability in this scenario; see \textbf{Section \ref{sec:unsteerability}}.

Similarly to the case of behaviours, where our allowed transformations preserved the local set, the LOSR transformations preserve the set of unsteerable assemblages, i.e.~if $\RR=\{P(a|x)\rho_{a|x}\}_{a,x}$ is unsteerable in the $A|B$ split, then the assemblage $\RR'=\{P'(a_0,a_1|x_0,x_1)\rho'_{a_0,a_1|x_0,x_1}\}_{a_0,a_1,x_0,x_1} \coloneqq \mathcal{M}(\RR)$ will also be unsteerable in the $A_0A_1|B_1B_2$ partition \cite{zjawin2023quantifyingepr}. Likewise, the set of transformations that preserves the unsteerable assemblages is larger than the set of LOSR transformations. In particular, a relevant class of operations that preserve the set of unsteerable assemblages are the one-way LOCC transformations (1W-LOCC) \cite{PhysRevX.5.041008}. These transformations require communication. However, communication occurs only from Bob to Alice, which prevents steerability from being created. Both sets of free transformations, 1W-LOCC and LOSR, are natural free operations for the steering scenario. However, in this work, we will focus solely on LOSR transformations.

\subsection{Unsteerability in the broadcasting scenario} \label{sec:unsteerability}

In the previous section, we have highlighted that LOSR transformations for assemblages preserve unsteerability from the bipartite $A|B$ partition into the quadripartite $A_0A_1|B_0B_1$ partition. 

However, analogously to the broadcasting scenario of behaviours, (see \textbf{Section \ref{subsection:NL in broadcasting scenario}}), \textbf{Definition \ref{def:LOSRa}} does not have any requirements on the marginals given an unsteerable input. Yet, well-defined marginals are needed if we want to meaningfully discuss broadcasting. To this end, we introduce the set of \emph{Unsteerable Realistic Non-Signalling assemblages}. This set consists of assemblages that are unsteerable in the partition $A_0A_1|B_0B_1$, and, in addition, the marginals have non-signalling components between $A_0$ and $A_1$.  More formally,
\begin{defi}\label{def:urns}
    An assemblage $\ZZ\!=\!\{Q(\bm{a}|\bm{x})\sigma_{\bm{a}|\bm{x}}\}_{\bm{a},\bm{x}}$ from the (2,2,m,o,d) steering scenario is \emph{unsteerable realistic non-signalling}, if there exists a probability distribution $\{r(\lambda)\}_\lambda$ over an exogenous variable $\lambda$, together with a set of quantum states $\{\sigma(\lambda)\}_\lambda\subset \mathcal{D}(\mathcal{H}_{B_0}\otimes\mathcal{H}_{B_1})$ and a behaviour $\{Q(a_0,a_1|x_0,x_1,\lambda)\}_{a_0,a_1,x_0,x_1,\lambda}$ that is non-signalling in the $A_0|A_1$ partition for every $\lambda$, such that every element of $\ZZ$ can be decomposed as:
    \begin{equation}
    \begin{split}
	Q(a_0,a_1|x_0,x_1)\sigma_{a_0,a_1|x_0,x_1}=\\
    =\sum_{\lambda} r(\lambda) Q(a_0,a_1|x_0,x_1,\lambda)\sigma(\lambda)\text{.}
    \end{split}
    \end{equation}
    The set of all assemblages with such properties will be denoted by $\text{UR}_{ns}$.
    \label{Def.UnsteerableRealisticDecomp}
\end{defi}

Note that this means a quadripartite unsteerable, realistic, non-signalling behaviour is unsteerable in every possible partition of participants. This set will be considered unsteerable in the broadcasting scenario.

Having introduced the set of assemblages of interest --- the set $\mathrm{UR_{ns}}$--- and a divergence, we can introduce a relative entropy in the steering scenario. Such a definition was initially introduced in Ref.~\cite{PhysRevX.5.041008}. However, in Ref.~\cite{PhysRevX.5.041008}, the set of free operations used is 1W-LOCC, whereas in our work, we focus on LOSR transformations. This enables us to introduce a simplified version of the relative entropy introduced in Ref.~\cite{PhysRevX.5.041008}--- we refer to Ref.~\cite{PhysRevA.96.022332} for more details.

\begin{defi}
    The \emph{relative entropy of steering} of an assemblage $\RR$ from the $(2,2,m,o,d)$ steering scenario is defined as:
    \begin{equation}
    E_{\mathrm{UR}}\left(\RR\right)\coloneqq \inf_{\ZZ \in \mathrm{UR_{ns}}} S_{\mathrm{a}}(\RR||\ZZ), 
    \end{equation}
    where the infimum is taken over all possible assemblages $\ZZ$ from the $(2,2,m,o,d)$ steering scenario that admit an unsteerable realistic non-signalling decomposition.
\end{defi}

Note that, with a slight abuse of notation, we can extend the relative entropy of nonlocality to behaviours from the $(2,m,o)$-scenario as well by treating the local realistic non-signalling set ($\mathrm{LR_{ns}}$) as the local set. This respects the properties of the above defined entropy since the marginals of a $\mathrm{LR}_{ns}$ behaviour are local.

Note that, with a slight abuse of notation, we can extend the relative entropy of steering to assemblages from the $(1,1,m,o,d)$ steering scenario as well by treating the unsteerable realistic non-signalling set ($\mathrm{UR_{ns}}$) as the unsteerable set. This respects the properties of the above-defined entropy since the marginals of a $\mathrm{UR}_{ns}$ assemblage are unsteerable.

Given the above definition of non-steerability in the broadcasting scenario, an additional requirement for the LOSR transformations to be admissible as local transformations is for them to preserve this definition of non-steerability:

\begin{defi}
A transformation of assemblages is said to be a \textit{$\text{UR}_{ns}$-LOSR assemblage transformation} when $(i)$ it belongs to the set of LOSR assemblage transformations and $(ii)$ it transforms unsteerable assemblages into assemblages in $\text{UR}_{ns}$.
\end{defi}

\subsection{No-local-broadcasting theorem for steerable non-signalling assemblages}

With this framework in place, we now establish our main no-go result for assemblages. The proof of it is very similar to the proof of \textbf{Theorem \ref{Thm.ImpossibilityBroadcastBoxes}} as the argument is based on a contradiction. To this end, two key propositions are presented that closely parallel the ones presented for behaviours.

The first key ingredient in the proof relies on the contractivity of the relative entropy of steering under LOSR transformations. This fundamental property, originally established in Refs.~\cite{PhysRevX.5.041008, PhysRevA.96.022332}, ensures that any LOSR transformation cannot increase the relative entropy of steering. For completeness, we restate this result:
\begin{prop}\label{prop:ContractivityEntropyAssemblages}
    Let $\mathcal{M}$ be a LOSR transformation and let $\RR$ be an assemblage, then:
\begin{equation}
     E_{\mathrm{UR}}(\mathcal{M}(\RR)) \le  E_{\mathrm{UR}}(\RR).
\end{equation}
\end{prop}

The second key ingredient is that if $\RR'$ is a broadcast version of a steerable assemblage $\RR$, then its relative entropy of steering must be strictly greater than that of $\RR$. The proof of this result parallels \textbf{Proposition \ref{prop:EntropyDilatation}}, with the key distinction that, instead of relying on the chain rule, we employ Theorem 1 from Ref.~\cite{piani2009relative}.

More formally, we establish the following:
\begin{restatable}{prop}{propDilationEntropyAssemblages}\label{prop:DilationEntropyAssemblages}
    If $\RR'$ is a broadcast version of a steerable assemblage $\RR$, then:
      \begin{equation}
          E_{\mathrm{UR}}(\RR') >  E_{\mathrm{UR}}(\RR).
      \end{equation}
\end{restatable}
The proof of this proposition is detailed in \textbf{Appendix} \ref{sec:appendix-proof-prop:DilatationAssemblages}.

We can now establish the second main result of this
work.

\begin{restatable}{thm}{ThrNonBroadcastingAssemblages}[No-local-broadcasting of assemblages]\label{thm:NonBroadcasting_Assemblages}
It is impossible to locally broadcast any known steerable non-signalling assemblage using $\text{UR}_{ns}$-LOSR transformations.
\end{restatable}

\begin{proof}
 Let us suppose, by contradiction, that there is an $\mathrm{UR_{ns}}$-LOSR transformation $\mathcal{M}$ that broadcasts a known steerable assemblage $\RR$. Then, by \textbf{Proposition \ref{prop:ContractivityEntropyAssemblages}}, we have:
\begin{equation}\label{eq:ThmEntropyContractivity_Assemblages}
        E_{\mathrm{UR}}\left(\mathcal{M}(\RR)\right) \le  E_{\mathrm{UR}}(\RR).
    \end{equation}
On the other hand, as $\mathcal{M}(\RR)$ is a broadcast version of $\RR$, by \textbf{Proposition \ref{prop:DilationEntropyAssemblages}}:
\begin{equation}\label{eq:ThmEntropyDilation_Assemblages}
        E_{\mathrm{UR}}(\RR) <  E_{\mathrm{UR}}\left(\mathcal{M}(\RR)\right).
    \end{equation}
Therefore, by Equations \eqref{eq:ThmEntropyContractivity_Assemblages} and \eqref{eq:ThmEntropyDilation_Assemblages}:
    \begin{equation}
         E_{\mathrm{UR}}(\mathcal{M}(\RR)) <  E_{\mathrm{UR}}(\mathcal{M}(\RR)),
    \end{equation}
which is a contradiction. Thus, it is impossible to broadcast any known steerable assemblage by $\mathrm{UR_{ns}}$-LOSR transformations. 
\end{proof}

\section{Discussion}

In this work, we provided negative answers for the possibility of local broadcasting of either nonlocal behaviours or steerable assemblages by using fundamental properties of the relative entropy of nonlocality and the relative entropy of steering. Answering, therefore, the long-standing conjecture raised by Joshi, Grudka and Horodecki$^{\otimes 4}$.

Similarly to the quantum case, we concluded that the local copying of information from assemblages and behaviours is impossible, provided that the object in question possesses a non-classical resource such as steerability or nonlocality. Since non-steerable assemblages and local behaviours are locally broadcastable,  we may say that the failure of local broadcasting indicates the non-classicality of the correlations. 

As we alluded to in the introduction, the fact that we were able to prove alternative forms of the no-broadcasting theorem for other resources underscores the importance of the standard version of the theorem. The impossibility of copying unknown information is not distinctive of the formalism of quantum theory, but of a broader class of probabilistic models satisfying no-signalling criteria. The reference \cite{Zhu_2022} also dealt with the problem of no-cloning in EPR scenarios. However, the definitions of cloning and the allowed operations used there differ from those used in our manuscript.

Additionally, we want to emphasise that the proofs presented in this work only encompass the bipartite-to-quadripartite scenarios. Although this may sound restrictive, the authors believe that the ideas contained in this article are already enough for a generalisation.

Generalised versions of the local no-broadcasting theorem might also exist for any non-classical finite-dimensional probabilistic model satisfying a no-signalling criterion. It is also possible that such a generalised local no-broadcasting theorem is equivalent to the original generalised no-broadcasting theorem of Barnum, Barrett, Leifer, and Wilce. We leave those points as open questions to be investigated in the future.

Finally, although we have formulated our no-local-broadcasting result for assemblages in terms of LOSR transformations, it is worth noting that the impossibility persists if one considers the larger class of one-way LOCC operations (1W-LOCC). Indeed, the main ingredient used in the proof, Proposition~3, remains valid for 1W-LOCC transformations. Our focus on LOSR is motivated by the desire to maintain a close analogy with the behaviour-based results and by the fact that LOSR operations are commonly regarded as the natural free operations in the resource-theoretic framework of steering \cite{zjawin2023quantifyingepr}.

\vspace{0.5cm}

\begin{acknowledgments}
AS and CV are grateful to the International Institute of Physics, where the groundwork for this project was completed. CV and CD thank the HUN-REN Wigner RCP for the welcoming environment in which part of this project was developed.  This research was also supported by the São Paulo Research Foundation (Fapesp) under grant no.~2024/16657-3 and 2025/01058-0, by the Fetzer Franklin Fund of the John E.\ Fetzer Memorial Trust and by grant number FQXi-RFP-IPW-1905 from the Foundational Questions Institute and Fetzer Franklin Fund, a donor-advised fund of Silicon Valley Community Foundation. Financial support from the Brazilian agency Coordenação de Aperfeiçoamento de Pessoal de Nível Superior - CAPES is gratefully acknowledged. 
This work has also been supported by the Ministry of Innovation and Technology and the National Research, Development and Innovation Office (NKFIH) within the Quantum Information National Laboratory of Hungary and through
OTKA Grants FK 135220, K 124152, and K 124351. This work was supported by the Conselho Nacional de Desenvolvimento Científico e Tecnológico (CNPq), through a grant from the Conhecimento Brasil Program - Line 1 and Line 2. 
\end{acknowledgments}

\bibliography{references}
\begin{widetext}
\appendix

\section{Proof of Proposition \ref{prop:EntropyContractivity}}\label{sec:appendix-proof-prop:EntropyContractivity}
\propEntropyContractivity*
\begin{proof}
    By definition,
    \begin{equation}
        E_{\mathrm{LR}}(\mathcal{M}(\PP)) = \inf_{\QQ \in \mathrm{LR_{ns}}} S_{\mathrm{b}}(\mathcal{M}(\PP)|| \QQ)\text{.}
    \end{equation}
    
    Since $\mathcal{M}$ is $\mathrm{LR_{ns}}$ preserving, the image of $\mathcal{M}$ over the set of local behaviours is a subset of the $\mathrm{LR_{ns}}$ set. Therefore, if we take the infimum of the above equation over the set $\mathrm{Im}(\mathcal{M})$, we will get an upper bound for $E_{\mathrm{LR}}(\mathcal{M}(\PP))$. Indeed,
    \begin{equation}
    \begin{split}
        E_{\mathrm{LR}}(\mathcal{M}(\PP)) &\le \inf_{\QQ \in \mathrm{Im}(\mathcal{M})} S_{\mathrm{b}}(\mathcal{M}(\PP)|| \QQ)\\
        &= \inf_{\tilde{\QQ} \in \mathrm{LR_{ns}}} S_{\mathrm{b}}(\mathcal{M}(\PP)|| \mathcal{M}(\tilde{\QQ})).
    \end{split}
    \end{equation}
    Now, by the contractivity of $S_{\mathrm{b}}$ over LOSR transformations for non-signalling behaviours --- eq.~\eqref{eq:ContracitivityDivergence}, we have:
    \begin{equation}
    \begin{split}
        E_{\mathrm{LR}}(\mathcal{M}(\PP)) &\le \inf_{\tilde{\QQ} \in \mathrm{LR_{ns}}} S_{\mathrm{b}}(\PP||\tilde{\QQ})\\
        &= E_{\mathrm{LR}}(\PP).
    \end{split}
    \end{equation}
\end{proof}

\section{Proof of Proposition \ref{prop:EntropyDilatation}}\label{sec:appendix-proof-prop:EntropyDilatation}

In this appendix, we prove the following proposition:
\propEntropyDilation*

The proof uses the \textit{chain rule}, which we first introduced for simple probability distributions. To make it easier to follow, the parties involved in the different behaviours or probability distributions will be noted as lower indices throughout this appendix. Given two joint probability distributions $\{p_{A_0A_1}(a_0,a_1)\}_{a_0,a_1}$ and $\{q_{A_0A_1}(a_0,a_1)\}_{a_0,a_1}$ over the same input space, we define the conditional probability distributions $\{p_{A_1|A_0}(a_1|a_0)\}_{a_1}$ and $\{q_{A_1|A_0}(a_1|a_0)\}_{a_1}$ for each $a_0$ they can be defined for. An interesting identity relating the KL-divergence of the joint probability distribution with the KL-divergence of the marginal and the average KL-divergence of the conditional (whenever it makes sense) is given by:
\begin{equation}
    S(p_{A_0A_1}||q_{A_0A_1}) = S(p_{A_0}||q_{A_0}) + \sum_{a_0} p_{A_0}(a_0) S\left(p_{A_1|A_0}(\,.\,|a_0)||q_{A_1|A_0}(\,.\,|a_0)\right).
\end{equation}
This identity is known as the \textit{chain rule} \cite{kullback_information_1997, cover1999elements}. Note that if, for a given $\underline{a_0}$, the conditional probability $\{p_{A_1|A_0}(a_1|\underline{a_0})\}_{a_1}$ is undefined, then for all $a_1$ the well-defined corresponding contributions to the KL-divergence on both sides equate to zero and thus every $(\underline{a_0},a_1)$ index pair can simply be disregarded. However, if only the conditional probability $\{q_{A_1|A_0}(a_1|\underline{a_0})\}_{a_1}$ is undefined, then the well-defined corresponding contributions to the KL-divergence on both sides are infinite. In this case, the chain rule is vacuous.

Now we go on to adapt the chain rule to behaviours. Let $\PP_{A_0A_1B_0B_1}=\{P_{A_0A_1B_0B_1}(\bm{a},\bm{b}|\bm{x},\bm{y})\}_{\bm{a},\bm{b},\bm{x},\bm{y}}$ and $\QQ_{A_0A_1B_0B_1}=\{Q_{A_0A_1B_0B_1}(\bm{a},\bm{b}|\bm{x},\bm{y})\}_{\bm{a},\bm{b},\bm{x},\bm{y}}$ be non-signalling behaviours in the $A_0B_0|A_1B_1$ split over the same input and output spaces. For shorthand, let us introduce the indices $0=A_0B_0$ and $1=A_1B_1$. If we fix the inputs $\underline{\bm{x}},\underline{\bm{y}}$, then we can apply the chain rule for $S(\PP_{01}(.\,,.|\underline{\bm{x}},\underline{\bm{y}})||\QQ_{01}(.\,,.|\underline{\bm{x}},\underline{\bm{y}}))$ in the $A_0B_0|A_1B_1$ split, resulting in:
    \begin{equation}
    \begin{split}\label{eq:ChainRuleBox}
       S(\PP_{01}(.\,,.|\underline{\bm{x}},\underline{\bm{y}})||\QQ_{01}(.\,,.|\underline{\bm{x}},\underline{\bm{y}})) =& \;S(\PP_{0}(.\,,.|\underline{x_0}, \underline{y_0})||\QQ_{0}(.\,,.|\underline{x_0}, \underline{y_0}))\\
       &+  \sum_{a_0,b_0} P_{0}(a_0,b_0|x_0,y_0)S(\PP_{1|0}(.\,,.|\underline{\bm{x}},\underline{\bm{y}}, a_0, b_0)||\QQ_{1|0}(.\,,.|\underline{\bm{x}},\underline{\bm{y}}, a_0, b_0)),
    \end{split}
    \end{equation}
where in eq.~\eqref{eq:ChainRuleBox} we have used the non-signalling feature of $\PP_{01}$ and $\QQ_{01}$.

Notice how Equation \eqref{eq:ChainRuleBox} contains the seed to prove Equation \eqref{eq:major}: the left-hand side pertains to the KL-divergence of the broadcast version, while the right-hand side pertains to the KL-divergence of the original version plus conditionals. In order to prove \textbf{Proposition \ref{prop:EntropyDilatation}}, we shall prove that these additional terms related to the conditionals do not vanish when the original behaviour was nonlocal. The proof is based on five lemmas. \textbf{Lemma \ref{lemma:ConditionalDivergenceBiggerZero}} is the backbone of the proof of \textbf{Proposition \ref{prop:EntropyDilatation}} as it shows that a non-vanishing term always remains, and \textbf{Lemmas \ref{lemma:MarginalNL}} and \textbf{\ref{lemma:MarginalLR}} support the proof of \textbf{Lemma \ref{lemma:ConditionalDivergenceBiggerZero}}. Meanwhile, \textbf{Lemmas \ref{lemma:finiteness}} and \textbf{\ref{lemma:MinEntropy}} are about the existence of a concrete, well-behaved, $\mathrm{LR_{ns}}$ behaviour that attains the entropy of nonlocality so that we can make calculations with it.

The following lemma is the first supporting lemma for \textbf{Lemma \ref{lemma:ConditionalDivergenceBiggerZero}}.

\begin{lemma}\label{lemma:MarginalNL}
    Let $\PP_{01}$ be a non-signalling behaviour in the $A_0B_0|A_1B_1$ split such that the marginal $\PP_{1}$ is nonlocal in the $A_1|B_1$ split. Then, for every input $x_0,y_0$, there exist outputs $a_0,b_0$ such that $P_{0}(a_0, b_0|x_0,y_0) \neq 0$ and the associated conditional behaviour $\PP_{1|0}(a_0,b_0,x_0,y_0)\coloneqq\{P_{1|0}(a_1, b_1|\bm{x},\bm{y}, a_0, b_0)\}_{a_1,b_1,x_1,y_1}$ is nonlocal in the $A_1|B_1$ split. 
\end{lemma}

In the following, the inputs $a_0,b_0,x_0,y_0$ for the associated conditional behaviour $\PP_{1|0}$ will sometimes be omitted.

\begin{proof}
    First, note that given inputs $x_0,y_0$ there exists at least one set of outputs $a_0,b_0$ such that $P_0(a_0,b_0|x_0,y_0)\neq 0$, otherwise we would not get a probability distribution out of $\PP_{01}$. Let us fix the inputs $x_0,y_0$ and choose a set of outputs $a_0,b_0$ such that  $P_0(a_0,b_0|x_0,y_0)\neq 0$. By the definition of a conditional probability distribution, we have:
    \begin{equation}
        P_{1|0}(a_1,b_1|\bm{x},\bm{y},a_0,b_0) \coloneqq \frac{P_{01}(a_0,a_1, b_0,b_1|\bm{x},\bm{y})}{P_{0}(a_0, b_0|x_0,y_0)}\text{.}
    \end{equation}
    Thus,
    \begin{equation}
        P_{01}(a_0,a_1, b_0,b_1|\bm{x},\bm{y}) = P_{0}(a_0, b_0|x_0,y_0)P_{1|0}(a_1,b_1|\bm{x},\bm{y},a_0,b_0).
    \end{equation}
    In this way, summing the above equation over the indices $a_0,b_0$, we have, using the non-signalling property of $\PP_{01}$:
    \begin{equation}
        P_{1}(a_1, b_1|x_1,y_1) = \sum_{a_0,b_0} P_{01}(a_0,a_1, b_0,b_1|\bm{x},\bm{y})  = \sum_{a_0,b_0} P_{0}(a_0, b_0|x_0,y_0)P_{1|0}(a_1,b_1|\bm{x},\bm{y},a_0,b_0).
    \end{equation}
    Note that when the sum hits an $a_0,b_0$ pair for which the conditional is not defined, both $P_{01}(a_0,a_1, b_0,b_1|\bm{x},\bm{y})$ and $P_{0}(a_0, b_0|x_0,y_0)$ are zero, thus such output pairs can be disregarded without loss of generality.
    
    Given the inputs $x_0, y_0$, let us suppose, for contradiction, that for every output $a_0, b_0$ such that $P_{0}(a_0, b_0|x_0,y_0) \neq 0$ we have that $\PP_{1|0}$ is local, but $\PP_{1}$ is nonlocal in the $A_1|B_1$ split. Let us use the notation $w_0\coloneqq (a_0,b_0,x_0,y_0)$. Thus, for each output $a_0, b_0$ such that $P_{0}(a_0, b_0|x_0,y_0) \neq 0$ there exists an exogenous variable $\lambda_{w_0}$ and probability distributions $\{r(\lambda_{w_0}|w_0)\}_{\lambda_{w_0}}$, $\{P_{A_1}(a_1|x_1,\lambda_{w_0},w_0)\}_{a_1,x_1}$ and $\{P_{B_1}(b_1|y_1,\lambda_{w_0},w_0)\}_{b_1,y_1}$ such that for all $a_1,b_1,x_1,y_1$:
   \begin{equation}
       P_{1|0}(a_1,b_1|x_1,y_1,w_0)\coloneqq P_{1|0}(a_1,b_1|\bm{x},\bm{y},a_0,b_0) = \sum_{\lambda_{w_0}} r(\lambda_{w_0}|w_0) P_{A_1}(a_1|x_1, \lambda_{w_0},w_0) P_{B_1}(b_1|y_1,\lambda_{w_0},w_0).
   \end{equation}
   Thus,
   \begin{equation}
   \begin{split}
        P_{1}(a_1, b_1|x_1,y_1) &= \sum_{a_0,b_0} P_{0}(a_0, b_0|x_0,y_0)P_{1|0}(a_1,b_1|\bm{x},\bm{y},a_0,b_0) \\
        &= \sum_{a_0,b_0}  P_{0}(a_0, b_0|x_0,y_0)\sum_{\lambda_{w_0}} r(\lambda_{w_0}|w_0) P_{A_1}(a_1|x_1, \lambda_{w_0},w_0) P_{B_1}(b_1|y_1,\lambda_{w_0},w_0).
    \end{split}
   \end{equation}
   
    Let $\lambda' \coloneqq (a_0, b_0,\lambda_{w_0})$ be a new exogenous variable, and $r'(\lambda'|x_0,y_0) \coloneqq P_0(a_0, b_0|x_0,y_0) r(\lambda_{w_0}|w_0)$. Then it is easy to see that $\{r'(\lambda'|x_0,y_0)\}_{\lambda'}$ is a probability distribution over the variable $\lambda'$, i.e.~$r'(\lambda'|x_0,y_0) \ge 0$ for all $\lambda'$ and $\sum_{\lambda'} r(\lambda'|x_0,y_0) = 1$. We also have:
    \begin{equation}
        P_1(a_1, b_1|x_1,y_1) = \sum_{\lambda'} r'(\lambda'|x_0,y_0)P_{A_1}(a_1|x_1,\lambda',x_0,y_0) P_{B_1}(b_1|y_1,\lambda',x_0,y_0). 
    \end{equation}

    Therefore, $\PP_1$ is local, which is in contradiction with the hypothesis. Given the no-signalling property of $\PP_{01}$, the left-hand side is not dependent on $x_0,y_0$ while the right-hand side is dependent. This means that the above equation must be true for all inputs $x_0,y_0$. Thus, the conclusion is that for any inputs $x_0,y_0$, there has to be at least one output $a_0,b_0$, such that $\PP_{1|0}$ is nonlocal in the $A_1|B_1$ split. 
\end{proof}

The following lemma is the second supporting lemma for \textbf{Lemma \ref{lemma:ConditionalDivergenceBiggerZero}}.

\begin{lemma}\label{lemma:MarginalLR}
    Let $\QQ_{01}$ be a $LR_{ns}$ behaviour. Then for every choice $a_0,b_0,x_0,y_0$ such that $Q_{0}(a_0,b_0|x_0,y_0)\neq 0$, we have that the related conditional behaviour $\QQ_{1|0}$ is local in the $A_1|B_1$ split.
\end{lemma}
\begin{proof}
    Let us start by fixing $w_0\coloneqq (a_0,b_0,x_0,y_0)$ such that $Q_{0}(a_0,b_0|x_0,y_0)\neq 0$. We have, by definition:
    \begin{equation}
        Q_{1|0}(a_1,b_1|\bm{x},\bm{y},a_0,b_0) \coloneqq \frac{Q_{01}(a_0,a_1, b_0,b_1|\bm{x},\bm{y})}{Q_{0}(a_0, b_0|x_0,y_0)}.
    \end{equation}

    We shall show that this is local in the $A_1|B_1$ split by exploiting the $\mathrm{LR_{ns}}$ property of $\QQ_{01}$. As $\QQ_{01}$ is a $\mathrm{LR_{ns}}$ behaviour, we have an exogenous variable $\lambda$, a probability distribution $\{r(\lambda)\}_{\lambda}$ and behaviours $\{\QQ_{A}(\lambda)\}_{\lambda}$ and $\{\QQ_{B}(\lambda)\}_{\lambda}$ such that for all inputs and outputs:
    \begin{equation}
        Q_{01}(\bm{a},\bm{b}|\bm{x},\bm{y}) = \sum_{\lambda} r(\lambda)Q_{A}(a_0, a_1|x_0,x_1,\lambda)Q_{B}(b_0, b_1|y_0,y_1,\lambda),
    \end{equation}
    where $\QQ_{A}(\lambda)$ and $\QQ_{B}(\lambda)$ are non-signalling in, respectively, the $A_0|A_1$ and $B_0|B_1$ splits. The requirement that $Q_{0}(a_0,b_0|x_0,y_0)\neq 0$ gives:
    \begin{equation}
    \begin{split}
        Q_{0}(a_0,b_0|x_0,y_0)=\sum_{a_1,b_1}Q_{01}(\bm{a},\bm{b}|\bm{x},\bm{y}) &= \sum_{a_1,b_1}\sum_{\lambda} r(\lambda)Q_{A}(a_0, a_1|x_0,x_1,\lambda)Q_{B}(b_0, b_1|y_0,y_1,\lambda)\\
        &=\sum_{\lambda} r(\lambda) \left(\sum_{a_1}Q_{A}(a_0, a_1|x_0,x_1,\lambda)\right)\left(\sum_{b_1}Q_{B}(b_0, b_1|y_0,y_1,\lambda)\right)\\
        &=\sum_{\lambda} r(\lambda) Q_{A_0}(a_0|x_0,\lambda) Q_{B_0}(b_0|y_0,\lambda)\neq 0\text{,}
    \end{split}
    \end{equation}
    where we have used the non-signalling property of $\QQ_{A}(\lambda)$ and $\QQ_{B}(\lambda)$ to introduce $\QQ_{A_0}(\lambda)$ and $\QQ_{B_0}(\lambda)$.

    The above equation tells us that for some $\lambda$ neither of the quantities $r(\lambda), Q_{A_0}(a_0|x_0,\lambda) $ and $Q_{B_0}(b_0|y_0,\lambda)$ is zero. Let us restrict ourselves to these $\lambda$ variables that we shall call $\lambda_{w_0}$. Thus, for the purposes of the fixed inputs and outputs $w_0\coloneqq (a_0,b_0,x_0,y_0)$, it is true that:
    \begin{equation}
        Q_{01}(\bm{a},\bm{b}|\bm{x},\bm{y}) = \sum_{\lambda_{w_0}} r(\lambda_{w_0})Q_{A}(a_0, a_1|x_0,x_1,\lambda_{w_0})Q_{B,}(b_0, b_1|y_0,y_1,\lambda_{w_0}),
    \end{equation}
    since for any $\lambda$, such that either the element $Q_{A_0}(a_0|x_0,\lambda)$ or the element $Q_{B_0}(b_0|y_0,\lambda)$ is zero, the corresponding quantity $Q_{A}(a_0, a_1|x_0,x_1,\lambda_{w_0})$ or $Q_{B}(b_0, b_1|y_0,y_1,\lambda_{w_0})$ is also zero for all free variables.

    Given that $Q_{A_0}(a_0|x_0,\lambda_{w_0})\neq 0$ and $Q_{B_0}(b_0|y_0,\lambda_{w_0})\neq 0$ for every $\lambda_{w_0}$ we can define the appropriate conditional behaviours to get:
    \begin{equation}
        Q_{A}(a_0, a_1|x_0,x_1,\lambda_{w_0}) = Q_{A_{1|0}}(a_1|x_0,x_1,a_0,\lambda_{w_0})Q_{A_0}(a_0|x_0,\lambda_{w_0}),
    \end{equation}
    and 
    \begin{equation}
        Q_{B}(b_0, b_1|y_0,y_1,\lambda_{w_0}) = Q_{B_{1|0}}(b_1|y_0,y_1,b_0,\lambda_{w_0})Q_{B_0}(b_0|y_0,\lambda_{w_0}).
    \end{equation}
    Thus, for the fixed $w_0\coloneqq (a_0,b_0,x_0,y_0)$ we can write:
    \begin{equation}
        Q_{01}(\bm{a},\bm{b}|\bm{x},\bm{y}) = \sum_{\lambda_{w_0}} r(\lambda_{w_0})Q_{A_{1|0}}(a_1|x_0,x_1,a_0,\lambda_{w_0})Q_{A_{0}}(a_0|x_0,\lambda_{w_0}) Q_{B_{1|0}}(b_1|y_0,y_1,b_0,\lambda_{w_0})Q_{B_0}(b_0|y_0,\lambda_{w_0})\text{.}
    \end{equation}
    Substituting this into the definition of $\QQ_{1|0}$:
    \begin{equation}
        Q_{1|0}(a_1,\!b_1|\bm{x},\!\bm{y},\!a_0,\!b_0)\! = \!\sum_{\lambda_{w_0}} \!\frac{r(\lambda_{w_0})Q_{A_0}\!(a_0|x_0,\!\lambda_{w_0})Q_{B_0}\!(b_0|y_0,\!\lambda_{w_0})}{Q_{0}(a_0, b_0|x_0,y_0)}Q_{A_{1|0}}(a_1|x_0,\!x_1,\!a_0,\!\lambda_{w_0}) Q_{B_{1|0}}(b_1|y_0,\!y_1,\!b_0,\!\lambda_{w_0})\text{.}
    \end{equation}
    We shall define the following new probability distribution:
    \begin{equation}
        r_{w_0}(\lambda_{w_0})\coloneqq \frac{r(\lambda_{w_0})Q_{A_0}(a_0|x_0,\lambda_{w_0})Q_{B_0}(b_0|y_0,\lambda_{w_0})}{Q_{0}(a_0, b_0|x_0,y_0)}.
    \end{equation}
    
    It follows trivially that $r_{w_0}(\lambda_{w_0})> 0$ and that $\sum_{\lambda_{w_0}}r_{w_0}(\lambda_{w_0})=1$. Thus, we have found for fixed variables $w_0=(a_0,b_0,x_0,y_0)$ a set of exogenous variables $\lambda_{w_0}$, a probability distribution $\{r_{w_0}(\lambda_{w_0})\}_{\lambda_{w_0}}$ and behaviours $\{\QQ_{A_{1|0}}(\lambda_{w_0})\}_{\lambda_{w_0}}$ and $\{\QQ_{B_{1|0}}(\lambda_{w_0})\}_{\lambda_{w_0}}$ such that for all $a_1,b_1,x_1,y_1$:
    \begin{equation}
        Q_{1|0}(a_1,b_1|\bm{x},\bm{y},a_0,b_0) = \sum_{\lambda_{w_0}} r_{w_0}(\lambda_{w_0})Q_{A_{1|0}}(a_1|x_0,x_1,a_0,\lambda_{w_0}) Q_{B_{1|0}}(b_1|y_0,y_1,b_0,\lambda_{w_0})\text{,}
    \end{equation}
    which means that $\QQ_{1|0}$ is local in the $A_1|B_1$ split for the given $a_0,b_0,x_0,y_0$.     
\end{proof}

The following lemma shows that a non-vanishing term always remains in the chain rule.

\begin{lemma}\label{lemma:ConditionalDivergenceBiggerZero}
    Let $\PP_{01}=\{P_{01}(\bm{a},\bm{b}|\bm{x},\bm{y})\}_{\bm{a},\bm{b},\bm{x},\bm{y}}$ be a non-signalling behaviour in the $A_0B_0|A_1B_1$ split such that $\PP_1=\{P_{1}(a_1,b_1|x_1,y_1)\}_{a_1,b_1,x_1,y_1}$ is nonlocal in the $A_1|B_1$ split. Let $\QQ_{01}=\{Q_{01}(\bm{a},\bm{b}|\bm{x},\bm{y})\}_{\bm{a},\bm{b},\bm{x},\bm{y}}$ be a $LR_{ns}$ behaviour such that $\mathrm{supp}(\PP_{01})\subseteq \mathrm{supp}(\QQ_{01})$. Then we have, for every $x_0,y_0$:
    \begin{equation}
        \max_{x_1,y_1} \sum_{a_0,b_0}P_0(a_0,b_0|x_0,y_0)S(\PP_{1|0}(., .|\bm{x},\bm{y}, a_0, b_0)||\QQ_{1|0}(., .|\bm{x},\bm{y}, a_0, b_0)) > 0,
    \end{equation}
    for every $x_0, y_0$.
\end{lemma}
\begin{proof}
    Given $x_0, y_0$, by \textbf{Lemma \ref{lemma:MarginalNL}}, there exists $\underline{a_0},\underline{b_0}$ such that $P_{0}(\underline{a_0}, \underline{b_0}|x_0,y_0) \neq 0$ and the behaviour $\PP_{1|0}=\{P_{1|0}(a_1, b_1|\bm{x},\bm{y}, \underline{a_0}, \underline{b_0})\}_{a_1,b_1,x_1,y_1}$ is nonlocal in the $A_1|B_1$ split. Since $\mathrm{supp}(\PP_{01})\subseteq \mathrm{supp}(\QQ_{01})$, we necessarily have that $Q_{0}(\underline{a_0}, \underline{b_0}|x_0,y_0) \neq 0$. However, then by \textbf{Lemma \ref{lemma:MarginalLR}}, the behaviour $\QQ_{1|0}=\{Q_{1|0}(a_1, b_1|\bm{x},\bm{y}, \underline{a_0}, \underline{b_0})\}_{a_1,b_1,x_1,y_1}$ is local in the $A_1|B_1$ split. Therefore:
    \begin{equation}
    \begin{split}
        & \max_{x_1,y_1} \sum_{a_0,b_0}P_{0}(a_0,b_0|x_0,y_0)S(\PP_{1|0}(., .|\bm{x},\bm{y}, a_0, b_0)||\QQ_{1|0}(., .|\bm{x},\bm{y}, a_0, b_0))\\
        &\ge\max_{x_1,y_1}P_{0}(\underline{a_0},\underline{b_0}|x_0,y_0)S(\PP_{1|0}(., .|\bm{x},\bm{y}, \underline{a_0}, \underline{b_0})||\QQ_{1|0}(., .|\bm{x},\bm{y}, \underline{a_0}, \underline{b_0}))\text{.}
    \end{split}
    \end{equation}
    But, as behaviours, $\{P_{1|0}(a_1, b_1|\bm{x},\bm{y}, \underline{a_0}, \underline{b_0})\}_{a_1,b_1,x_1,y_1}$ and $\{Q_{1|0}(a_1, b_1|\bm{x},\bm{y}, \underline{a_0}, \underline{b_0})\}_{a_1,b_1,x_1,y_1}$ are different, since the first is nonlocal while the second is local in the $A_1|B_1$ split. Therefore, 
    \begin{equation}
        \max_{x_1,y_1}S(\PP_{1|0}(., .|\bm{x},\bm{y}, \underline{a_0}, \underline{b_0})||\QQ_{1|0}(., .|\bm{x},\bm{y}, \underline{a_0}, \underline{b_0}))>0.
    \end{equation}
    Thus,
    \begin{equation}
        \max_{x_1,y_1} \sum_{a_0,b_0}P_{0}(a_0,b_0|x_0,y_0)S(\PP_{1|0}(., .|\bm{x},\bm{y}, a_0, b_0)||\QQ_{1|0}(., .|\bm{x},\bm{y}, a_0, b_0))>0.
    \end{equation}
\end{proof}

The following lemma guarantees the finiteness of the relative entropy of nonlocality, and thus the finiteness of the terms in the chain rule.

\begin{lemma}\label{lemma:finiteness}
    Let $\PP$ be any behaviour from the $(4,m,o)$-scenario. Then the relative entropy of nonlocality of $\PP$ is always finite.
\end{lemma}

\begin{proof}
    By definition, we have:
    \begin{equation}
        E_{\mathrm{LR}}(\PP)=\inf_{\QQ\in\mathrm{LR_{ns}}} S_\mathrm{b}(\PP||\QQ)\text{.}
    \end{equation}
    In general, the relative entropy of behaviours is always non-negative but can be infinite. However, it is easy to see that it must be finite, as one can always choose a special uniform element $\hat{\QQ}$ from $\mathrm{LR_{ns}}$:
    \begin{equation}
       \hat{Q}(\bm{a},\bm{b}|\bm{x},\bm{y})\coloneqq \frac{1}{o^4}\text{,}
    \end{equation}
    where $o$ is the number of outputs of the behaviour. It is easy to check that this is an $\mathrm{LR_{ns}}$ behaviour as it simply factorises into the product of individual boxes for the parties $A_0,A_1,B_0$ and $B_1$ all with uniform outputs regardless of inputs. Furthermore:
    \begin{equation}
        S_b(\PP||\hat{\QQ})=\max_{\bm{x},\bm{y}}\left[\log(o^4)-H(P(.,.||\bm{x},\bm{y}))\right]\leq \log(o^4)\text{,}
    \end{equation}
    where $H$ is the Shannon entropy of the probability distribution $\{P(\bm{a},\bm{b}|\bm{x},\bm{y})\}_{\bm{a},\bm{b}}$, which takes values between 0 and $\log(o^4)$.
    Finally, by definition:
    \begin{equation}
       E_{\mathrm{LR}}(\PP)=\inf_{\QQ\in\mathrm{LR_{ns}}} S_\mathrm{b}(\PP||\QQ) \leq  S_b(\PP||\hat{\QQ})\leq \log(o^4)\text{.}
    \end{equation}
\end{proof}

The following lemma guarantees for all behaviours the existence of a local realistic non-signalling behaviour with which the relative entropy of behaviours reaches its infimum.

\begin{lemma}\label{lemma:MinEntropy}
For any behaviour $\PP$, there is always a behaviour $\bar{\QQ}\in \mathrm{LR_{ns}}$ such that $\mathrm{supp}(\PP)\subseteq \mathrm{supp}(\bar{\QQ})$ and $E_{\mathrm{LR}}(\PP) = S_{\mathrm{b}}(\PP||\bar{\QQ})$. In other words, the infimum in the definition of the relative entropy of nonlocality is always achieved with a behaviour $\bar{\QQ}$ whose support is a superset of that of $\PP$.
\end{lemma}
\begin{proof}
    It can be easily seen that the set of $\mathrm{LR_{ns}}$ behaviours is a polytope \cite{pironio2005lifting}, therefore it is a compact set. On the other hand, fixing a behaviour $\PP$, let us define $f(\QQ) := S_{\mathrm{b}}(\PP||\QQ)$. This is lower semi-continuous on the set of $\mathrm{LR_{ns}}$ behaviours \cite{dupuis1997weak}. Therefore, by the extended Bolzano–Weierstrass theorem, the function $f$ has a minimum in $\mathrm{LR_{ns}}$.

    Furthermore, we have from \textbf{Lemma \ref{lemma:finiteness}} that the entropy of nonlocality of any behaviour is upper-bounded. This necessarily means that any behaviour $\bar{\QQ}$ such that $E_{\mathrm{LR}}(\PP) = S_{\mathrm{b}}(\PP||\bar{\QQ})$ must have a support that is a superset of that of $\PP$ otherwise the entropy of nonlocality would give infinity, which is a contradiction.
\end{proof}

Given these results, let us now move on to the main goal of this section, proving \textbf{Proposition \ref{prop:EntropyDilatation}}.
\propEntropyDilation*
\begin{proof}
Let $\PP'=\PP'_{01}\eqqcolon\{P'_{01}(\bm{a},\bm{b}|\bm{x},\bm{y})\}_{\bm{a},\bm{b},\bm{x},\bm{y}}$ be the broadcast version of the behaviour $\PP\eqqcolon\{P(a,b|x,y)\}_{a,b,x,y}$ which is non-local in the $A|B$ split. Then, by definition,
\begin{equation}
    E_{\mathrm{LR}}(\PP'_{01})=  \inf_{\QQ_{01} \in \mathrm{LR_{ns}}} S_{\mathrm{b}}(\PP'_{01}||\QQ_{01})= S_{\mathrm{b}}(\PP'_{01}||\bar{\QQ}_{01})\text{,}
\end{equation}
where we have used \textbf{Lemma \ref{lemma:MinEntropy}} in the last equality. Note also that $\mathrm{supp}(\PP_{01})\subseteq \mathrm{supp}(\bar{\QQ}_{01})$. Then,
\begin{equation}
    E_{\mathrm{LR}}(\PP'_{01}) =  S_{\mathrm{b}}(\PP'_{01}||\bar{\QQ}_{01})= \max_{\bm{x},\bm{y}} S(\PP'_{01}(.\,,.|\bm{x},\bm{y})||\bar{\QQ}_{01}(.\,,.|\bm{x},\bm{y})).
\end{equation}
Using now the chain rule \eqref{eq:ChainRuleBox} for each $\bm{x},\bm{y}$ pair, we have:
\begin{align}
    &E_{\mathrm{LR}}(\PP'_{01})=\\
    &=  \max_{\bm{x},\bm{y}} \left[S(\PP'_{0}(., .|x_0, y_0)||\bar{\QQ}_{0}(., .|x_0, y_0)) + \sum_{a_0, b_0} P'_0(a_0,b_0|x_0,y_0)S(\PP'_{1|0}(., .|\bm{x},\bm{y}, a_0, b_0)||\bar{\QQ}_{1|0}(., .|\bm{x},\bm{y}, a_0, b_0))\right]\nonumber\\
    &= \max_{x_0,y_0} \max_{x_1,y_1}\left[S(\PP'_{0}(., .|x_0, y_0)||\bar{\QQ}_{0}(., .|x_0, y_0)) + \sum_{a_0, b_0} P'_0(a_0,b_0|x_0,y_0)S(\PP'_{1|0}(., .|\bm{x},\bm{y}, a_0, b_0)||\bar{\QQ}_{1|0}(., .|\bm{x},\bm{y}, a_0, b_0))\right]\nonumber\\
    &= \max_{x_0,y_0}\left[S(\PP'_{0}(., .|x_0, y_0)||\bar{\QQ}_{0}(., .|x_0, y_0)) + \max_{x_1,y_1}\sum_{a_0, b_0} P'_0(a_0,b_0|x_0,y_0)S(\PP'_{1|0}(., .|\bm{x},\bm{y}, a_0, b_0)||\bar{\QQ}_{1|0}(., .|\bm{x},\bm{y}, a_0, b_0))\right]\nonumber.
\end{align}
Applying \textbf{Lemma \ref{lemma:ConditionalDivergenceBiggerZero}} we get,
\begin{equation}
\begin{split}
    E_{\mathrm{LR}}(\PP'_{01}) &>  \max_{x_0,y_0}S(\PP'_{0}(., .|x_0, y_0)||\bar{\QQ}_{0}(., .|x_0, y_0))\\
    &= S_{\mathrm{b}}(\PP'_0||\bar{\QQ}_0).
\end{split}
\end{equation}
It is trivial to see that $\bar{\QQ}_0$ is a local behaviour, therefore:
\begin{equation}
\begin{split}
    E_{\mathrm{LR}}(\PP'_{01}) &> \inf_{\QQ \in \mathrm{LR_{ns}}}S_{\mathrm{b}}(\PP'_0||\QQ)\\
    &= E_{\mathrm{LR}}(\PP),
\end{split}
\end{equation}
where in the last line we have used that if $\PP'_{01}$ is a broadcast version of $\PP$, then $\PP'_0 = \PP$.
\end{proof}

\section{Proof of Proposition \ref{prop:DilationEntropyAssemblages}}\label{sec:appendix-proof-prop:DilatationAssemblages}

In this appendix, we prove the following proposition:
\propDilationEntropyAssemblages*

We continue to use the previously introduced $0$ and $1$ indices to denote either $A_iB_i$ ($i=0,1$) or simply $A_i$ or $B_i$, whichever is appropriate.

Whereas in the previous appendix we used the chain rule to prove the previous inequality, here we shall use an inequality by Piani \cite{piani2009relative}. Let us first introduce the concept of a measurement map.
\begin{defi}
    Given a POVM $\{\Pi_k\}_{k=1}^{n}\subset \mathcal{L}(\mathcal{H})$ the measurement map associated to it is defined as $\mathcal{N}: \mathcal{L}(\mathcal{H}) \to \mathcal{L}(\mathbb{C}^{n})$, such that
    \begin{equation}
        \mathcal{N}(X) \coloneqq \sum_{k=1}^{n} \trace(\Pi_kX) \ketbra{k}{k},
    \end{equation}
    where $X\in\mathcal{L}(\mathcal{H})$ and $\{\ket{k}\}_{k=1}^{n}$ is an orthonormal basis in $\mathbb{C}^{n}$.
\end{defi}
Let $\rho_{W\!Z},\sigma_{W\!Z}\in\mathcal{D}(\mathcal{H}_W\otimes\mathcal{H}_Z)$ be bipartite quantum states and let $\mathcal{N}: \mathcal{L}(\mathcal{H}_W) \to \mathcal{L}(\mathbb{C}^{n})$ be a measurement map associated to a POVM $\{\Pi_k\}_{k=1}^{n}$ on $\mathcal{H}_W$. In Ref.~\cite{piani2009relative}, Piani showed the following inequality:
\begin{equation}\label{eq:Piani's_Result}
    S_{\mathrm{q}}(\rho_{W\!Z}||\sigma_{W\!Z}) \ge S_{\mathrm{q}}(\mathcal{N}(\rho_W)||\mathcal{N}(\sigma_W)) + S_{\mathrm{q}}\left(\rho_Z\middle|\middle| \sum_{k=1}^{n} \alpha_{k} \sigma_{Z}^{k}\right),
\end{equation}
where $\rho_W,\sigma_{W}$, and $\rho_Z$ are the appropriate marginals, $\alpha_{k} \coloneqq \trace(\Pi_k \rho_W)$, and $\sigma_{Z}^{k} \coloneqq \frac{\trace_{W}((\Pi_{k} \otimes \mathbb{1}_Z) \sigma_{W\!Z})}{\trace((\Pi_{k} \otimes \mathbb{1}_Z) \sigma_{W\!Z})}$. Note that if $\alpha_k\neq 0$ and $\trace((\Pi_{k} \otimes \mathbb{1}_Z) \sigma_{W\!Z})=0$ the state $\sigma^k_{Z}$ cannot be defined. However, in this case, Piani's inequality is vacuous, as it can be shown that both the left-hand side and the first term on the right-hand side are equal to infinity. If $\alpha_k=\trace((\Pi_{k} \otimes \mathbb{1}_Z) \sigma_{W\!Z})=0$  then the corresponding $k$ index can be disregarded in Piani's inequality. It is also important to note that $\alpha_k$ is dependent on $\rho_W$ and $\sigma^k_Z$ is dependent on $\sigma_{W\!Z}$.

This inequality will form the basis of our proof in the following way: let us have two non-signalling assemblages $\RR_{01}\coloneqq\{P_{01}(a_0,a_1|x_0,x_1)\rho_{a_0,a_1|x_0,x_1}\}_{a_0,a_1,x_0,x_1}$ and $\ZZ_{01}\coloneqq\{Q_{01}(a_0,a_1|x_0,x_1)\sigma_{a_0,a_1|x_0,x_1}\}_{a_0,a_1,x_0,x_1}$ in the $(2,2,m,o,d)$ broadcasting scenario. Let $\{\pi_{01}(x_0,x_1)\}_{x_0,x_1}$ be a probability distribution for the choices of the parties $A_0$ and $A_1$. In order not to overload the notation, we will introduce the shorthand $W \coloneqq (E_0F_0B_0)$ and $Z\coloneqq (E_1F_1B_1)$. Thus, the CQ-states associated to the assemblages $\RR_{01}$ and $\ZZ_{01}$ are, respectively:
\begin{equation}
    \rho_{W\!Z} \coloneqq \underset{\substack{a_0,a_1\\x_0,x_1}}{\sum} \pi_{01}(x_0,x_1)P_{01}(a_0,a_1|x_0,x_1){\ketbra{x_0,x_1}{x_0,x_1} \otimes \ketbra{a_0,a_1}{a_0,a_1} \otimes \rho_{a_0,a_1|x_0,x_1}},
\end{equation}
and 
\begin{equation}
    \sigma_{W\!Z} \coloneqq \underset{\substack{a_0,a_1\\x_0,x_1}}{\sum} \pi_{01}(x_0,x_1)Q_{01}(a_0,a_1|x_0,x_1){\ketbra{x_0,x_1}{x_0,x_1} \otimes \ketbra{a_0,a_1}{a_0,a_1} \otimes \sigma_{a_0,a_1|x_0,x_1}},
\end{equation}
where we naturally have $\ket{x_0,x_1}=\ket{x_0}\otimes\ket{x_1}$ and similarly for other vectors.

Additionally, we will consider that $\ZZ_{01}$ is an $\mathrm{UR_{ns}}$ assemblage. Thus, there exists a probability distribution $\{r(\lambda)\}_\lambda$ over an exogenous variable $\lambda$, a set of quantum states $\{\sigma_{01}(\lambda)\}_{\lambda}\subset\mathcal{D}(\mathcal{H}_{B_0}\otimes\mathcal{H}_{B_1})$, and a behaviour $\{Q_{01}(a_0,a_1|x_0,x_1,\lambda)\}_{a_0,a_1,x_0,x_1,\lambda}$ such that we have:
\begin{equation}
    \sigma_{W\!Z} = \sum_{\lambda}\underset{\substack{a_0,a_1\\x_0,x_1}}{\sum} \pi_{01}(x_0,x_1)r(\lambda)Q_{01}(a_0,a_1|x_0,x_1,\lambda)\ketbra{x_0,x_1}{x_0,x_1} \otimes \ketbra{a_0,a_1}{a_0,a_1} \otimes \sigma_{01}(\lambda),
\end{equation}
where, for all $\lambda$, the behaviour $\{Q_{01}(a_0,a_1|x_0,x_1,\lambda)\}_{a_0,a_1,x_0,x_1}$ is non-signalling in the $A_0|A_1$ split. 

Now, let $\{\Gamma_i\}_{i=1}^{n} \subset \mathcal{L}(\mathcal{H}_{B_0})$ be an arbitrary Informationally Complete POVM (IC-POVM) on $\mathcal{H}_{B_0}$. Using this, we define the POVM $\{\Pi_{k}\}_k \subset \mathcal{L}(\mathcal{H}_W)=\mathcal{L}(\mathbb{C}^{m} \otimes \mathbb{C}^{o} \otimes \mathcal{H}_{B_0})$ as $\Pi_{k} \coloneqq \Pi_{x_0,a_0,i} = \ketbra{x_0}{x_0} \otimes \ketbra{a_0}{a_0} \otimes \Gamma_i$, where we have introduced the abbreviation $k\coloneqq(x_0,a_0,i)$. Note that $\mathbb{C}^m\cong\mathcal{H}_{E_0}$ and $\mathbb{C}^o\cong\mathcal{H}_{F_0}$. 
Let $\mathcal{N}: 
\mathcal{L}(\mathbb{C}^{m} \otimes \mathbb{C}^{o} \otimes \mathcal{H}_{B_0}) \to \mathcal{L}(\mathbb{C}^{m} \otimes \mathbb{C}^{o} \otimes  \mathbb{C}^{n})$ be the measurement map associated to the POVM $\{\Pi_{k}\}_k$, i.e. 
\begin{equation}
    \mathcal{N}(X) \coloneqq \sum_{x_0,a_0,i} \trace\big[(\ketbra{x_0}{x_0} \otimes \ketbra{a_0}{a_0} \otimes \Gamma_i)X\big] \ketbra{x_0}{x_0} \otimes \ketbra{a_0}{a_0} \otimes \ketbra{i}{i}, 
\end{equation}
for some $X\in\mathcal{L}(\mathbb{C}^{m} \otimes \mathbb{C}^{o} \otimes \mathcal{H}_{B_0})$. Then, by applying eq.~\eqref{eq:Piani's_Result} to the states $\rho_{W\!Z}$ and $\sigma_{W\!Z}$ using the measurement map $\mathcal{N}$, we have:
\begin{equation}\label{eq:piani_for_prop}
    S_{\mathrm{q}}(\rho_{W\!Z}||\sigma_{W\!Z}) \ge S_{\mathrm{q}}(\mathcal{N}(\rho_W)||\mathcal{N}(\sigma_W)) + S_{\mathrm{q}}\left(\rho_Z\middle|\middle| \sum_{k} \alpha_{k} \sigma_{Z}^{k}\right).
\end{equation}
Notice how the left-hand side already contains the building block for the KL-divergence of assemblage $\RR_{01}$ from $\ZZ_{01}$. Thus, Piani's inequality will be used in this manner throughout the appendix. To prove \textbf{Proposition \ref{prop:DilationEntropyAssemblages}}, we shall prove that the first term on the right-hand side of Equation \ref{eq:piani_for_prop} is greater than zero and that the second term is related to the relative entropy of steering of the marginal. In the following, we will often use independent probability distributions $\{\pi_{01}(x_0,x_1)\}_{x_0,x_1}=\{\pi_0(x_0)\pi_{1}(x_1)\}_{x_0,x_1}$ to achieve the desired results. Later on, we shall show that this does not lead to loss of generality.

Before going on to prove the proposition, we prove six auxiliary statements that will help us reach the desired conclusion. \textbf{Lemma \ref{lemma:Reduced_CQ-states}} gives a straightforward connection between the CQ-state of a multipartite assemblage and the CQ-state of its marginal. \textbf{Lemmas \ref{lemma:post_measurement_assemblages_unsteerable}} and \textbf{\ref{lemma:convex_sum_CQ_assemblages}} establish that the convex sum of states in Piani's inequality is unsteerable if the original assemblage was $\mathrm{UR_{ns}}$. \textbf{Lemma \ref{lemma:injectivity_ICPOVMs}} proves that assemblages can remain comparable even after using a measurement map. \textbf{Lemma \ref{lemma:finite_assemblage}} shows that, for our use, the terms in Piani's inequality are all finite, and thus the inequality is not vacuous. Finally, \textbf{Lemma \ref{lemma:analysis}} is a basic result from real analysis on the exchange of infima and suprema that we use during the proof of the main proposition.

The following lemma shows the connection between the CQ-state of a multipartite assemblage and the CQ-state of its marginal. 

\begin{lemma}\label{lemma:Reduced_CQ-states}
    Let $\RR_{01}=\{P_{01}(a_0,a_1|x_0,x_1)\rho_{a_0,a_1|x_0,x_1}\}_{a_0,a_1,x_0,x_1}$ be a non-signalling assemblage in the $0|1$ split, and let $\{\pi_{01}(x_0,x_1)\}_{x_0,x_1}$ be a probability distribution. If $\rho_{W\!Z}$ is the CQ-state associated to the pair $(\RR_{01}, \{\pi_{01}(x_0,x_1)\}_{x_0,x_1})$, then $\rho_W \coloneqq \trace_{Z}(\rho_{W\!Z})$ is the CQ-state associated to the marginal pair $(\RR_{0}, \{\pi_0(x_0)\}_{x_0})$.
\end{lemma}
\begin{proof}
By definition of the CQ-state,
\begin{equation}
    \rho_{W\!Z} \coloneqq \underset{\substack{a_0,a_1\\x_0,x_1}}{\sum} \pi_{01}(x_0,x_1)P_{01}(a_0,a_1|x_0,x_1)\ketbra{x_0,x_1}{x_0,x_1} \otimes \ketbra{a_0,a_1}{a_0,a_1} \otimes \rho_{a_0,a_1|x_0,x_1}.
\end{equation}
Thus,
\begin{align}
    \rho_{W} &= \trace_{Z}(\rho_{W\!Z})\nonumber\\
    &= \trace_{E_1F_1B_1}\left[ \underset{\substack{a_0,a_1\\x_0,x_1}}{\sum} \pi_{01}(x_0,x_1)P_{01}(a_0,a_1|x_0,x_1) \ketbra{x_0,x_1}{x_0,x_1} \otimes \ketbra{a_0,a_1}{a_0,a_1} \otimes \rho_{a_0,a_1|x_0,x_1}\right]\nonumber\\
    &=\underset{\substack{a_0,a_1\\x_0,x_1}}{\sum} \pi_{01}(x_0,x_1)P_{01}(a_0,a_1|x_0,x_1)\ketbra{x_0}{x_0} \otimes \ketbra{a_0}{a_0} \otimes \trace_{B_1}(\rho_{a_0,a_1|x_0,x_1})\\
    &= \sum_{a_0,x_0,x_1} \pi_{01}(x_0,x_1)\ketbra{x_0}{x_0} \otimes \ketbra{a_0}{a_0} \otimes \left(\sum_{a_1}P_{01}(a_0,a_1|x_0,x_1)\trace_{B_1}({\rho_{a_0,a_1|x_0,x_1}})\right)\text{.}\nonumber
\end{align}
As $\RR_{01}$ is a non-signalling assemblage in the $0|1$ split we have that $\sum_{a_1}P_{01}(a_0,a_1|x_0,x_1)\trace_{B_1}(\rho_{a_0,a_1|x_0,x_1}) = P_0(a_0|x_0)\rho_{a_0|x_0}$, which is, by Definition \ref{def:multipart_ns}, an element of the marginal assemblage $\RR_0=\{P_{0}(a_0|x_0)\rho_{a_0|x_0}\}_{a_0,x_0}$. Consequently,
\begin{equation}
\begin{split}
    \rho_{W} &=  \sum_{a_0,x_0} \left(\sum_{x_1}\pi_{01}(x_0,x_1)\right)P_0(a_0|x_0)\ketbra{x_0}{x_0} \otimes \ketbra{a_0}{a_0} \otimes \rho_{a_0|x_0}\\
    &=  \sum_{a_0,x_0} \pi_0(x_0)P_0(a_0|x_0)\ketbra{x_0}{x_0} \otimes \ketbra{a_0}{a_0} \otimes \rho_{a_0|x_0}.
\end{split}
\end{equation}
Therefore, $\rho_{W}$ is the CQ-state associated to the marginal pair $(\RR_{0}, \{\pi_0(x_0)\}_{x_0})$.
\end{proof}

The following lemma establishes that the marginal-like states appearing in Piani's inequality are related to unsteerable assemblages if the original assemblage was $\mathrm{UR_{ns}}$ and if the probability distribution was a product distribution.

\begin{lemma}\label{lemma:post_measurement_assemblages_unsteerable}
    Let $\ZZ_{01}=\{Q_{01}(a_0,a_1|x_0,x_1) \sigma_{a_0,a_1|x_0,x_1}\}_{a_0,a_1,x_0,x_1}$ be an $\mathrm{UR_{ns}}$ assemblage from the $(2,2,m,o,d)$ steering scenario, and let $\{\pi_{0}(x_0)\pi_{1}(x_1)\}_{x_0,x_1}$ be a product probability distribution leading to a CQ-state $\sigma_{W\!Z}$. Let $\{\Pi_k\}_{k}\coloneqq\{\ketbra{x_0}{x_0} \otimes \ketbra{a_0}{a_0} \otimes \Gamma_i\}_{x_0,a_0,i}\subset\mathcal{L}(\mathbb{C}^{m} \otimes \mathbb{C}^{o} \otimes \mathcal{H}_{B_0})$ be a POVM such that $\trace((\Pi_k\otimes\mathbb{1}_Z)\sigma_{W\!Z})\neq 0$ for all $k$. Then, for each $k=(x_0,a_0,i)$ the state $\sigma_{Z}^{k}\coloneqq\frac{\trace_{W}((\Pi_{k} \otimes \mathbb{1}_Z) \sigma_{W\!Z})}{\trace((\Pi_{k} \otimes \mathbb{1}_Z) \sigma_{W\!Z})}$ is a CQ-state associated to an unsteerable assemblage with the marginal probability distribution $\{\pi_1(x_1)\}_{x_1}$.
\end{lemma}
\begin{proof}
    Given a fixed $\underline{k} = (\underline{x_0},\underline{a_0},\underline{i})$, by definition we have
    \begin{equation}
        \sigma_{Z}^{\underline{k}} \coloneqq \frac{\trace_{W}((\Pi_{\underline{k}} \otimes \mathbb{1}_Z) \sigma_{W\!Z})}{\trace((\Pi_{\underline{k}} \otimes \mathbb{1}_Z) \sigma_{W\!Z})}.
    \end{equation}
    Using the fact that $\ZZ_{01}$ is an $\mathrm{UR_{ns}}$ assemblage, we have an exogenous variable $\lambda$, a probability distribution $\{r(\lambda)\}_{\lambda}$, a behaviour $\{Q_{01}(a_0,a_1|x_0,x_1,\lambda)\}_{a_0,a_1,x_0,x_1,\lambda}$ that is non-signalling in the $A_0|A_1$ split for each $\lambda$, and quantum states $\{\sigma_{01}(\lambda)\}_\lambda\subset\mathcal{D}(\mathcal{H}_{B_0}\otimes\mathcal{H}_{B_1})$, such that:
    \begin{align}
        &(\Pi_{\underline{k}} \otimes \mathbb{1}_Z) \sigma_{W\!Z} = (\Pi_{\underline{x_0},\underline{a_0},\underline{i}} \otimes \mathbb{1}_Z) \sigma_{W\!Z}\\
        &= \underset{\substack{a_0,a_1\\x_0,x_1}}{\sum} \!\pi_0(x_0)\pi_1(x_1)Q_{01}(a_0,a_1|x_0,x_1)(\ketbra{\underline{x_0}}{\underline{x_0}}\!\ketbra{x_0}{x_0}) \otimes \ketbra{x_1}{x_1} \otimes (\ketbra{\underline{a_0}}{\underline{a_0}}\!\ketbra{a_0}{a_0}) \otimes \ketbra{a_1}{a_1} \otimes (\Gamma_{\underline{i}}\otimes \mathbb{1}_{B_{1}})\sigma_{a_0,a_1|x_0,x_1}\nonumber\\
        &= \sum_{\lambda}\sum_{a_1,x_1} \pi_0(\underline{x_0})\pi_1(x_1)r(\lambda) Q_{01}(\underline{a_0},a_1|\underline{x_0},x_1,\lambda) \ketbra{\underline{x_0}}{\underline{x_0}} \otimes \ketbra{x_1}{x_1} \otimes \ketbra{\underline{a_0}}{\underline{a_0}} \otimes \ketbra{a_1}{a_1} \otimes (\Gamma_{\underline{i}}\otimes\mathbb{1}_{B_1}) \sigma_{01}(\lambda)\text{.}\nonumber
    \end{align}
    Thus,
    \begin{align}
        \trace_{W}\left[(\Pi_{\underline{k}} \otimes \mathbb{1}_Z) \sigma_{W\!Z}\right] &= \trace_{E_0F_0B_0}((\Pi_{\underline{x_0},\underline{a_0},\underline{i}} \otimes \mathbb{1}_Z) \sigma_{W\!Z})\\
        &= \sum_{\lambda}\sum_{a_1,x_1} \pi_0(\underline{x_0})\pi_1(x_1)r(\lambda) Q_{01}(\underline{a_0},a_1|\underline{x_0},x_1,\lambda) \ketbra{x_1}{x_1} \otimes \ketbra{a_1}{a_1} \otimes \trace_{B_0}((\Gamma_{\underline{i}}\otimes\mathbb{1}_{B_1}) \sigma_{01}(\lambda))\nonumber\\
        &= \sum_{\lambda}\sum_{a_1,x_1} \pi_0(\underline{x_0})\pi_1(x_1)r(\lambda) Q_{01}(\underline{a_0},a_1|\underline{x_0},x_1,\lambda) \ketbra{x_1}{x_1} \otimes \ketbra{a_1}{a_1} \otimes \tilde{\sigma}_{1}(\lambda,\underline{i}),\nonumber
    \end{align}
    where we have introduced the unnormalised quantum states $\{\tilde{\sigma}_{1}(\lambda,\underline{i}) \coloneqq \trace_{B_0}((\Gamma_{\underline{i}}\otimes\mathbb{1}_{B_1})\sigma_{01}(\lambda))\}_{\lambda}$. Note that if $\tilde{\sigma}_{1}(\lambda,\underline{i})$ were zero, then the corresponding $\lambda$ index can simply be skipped. Note also that there must be at least one $\lambda$ index such that $\tilde{\sigma}_{1}(\lambda,\underline{i})$ is nonzero given that $\trace((\Pi_{\underline{k}}\otimes\mathbb{1}_Z)\sigma_{W\!Z})\neq 0$.
    
    Given that for each $\lambda$ the behaviours $\{Q_{01}(a_0,a_1|x_0,x_1,\lambda)\}_{a_0,a_1,x_0,x_1}$ are non-signalling in the $A_0|A_1$ split, we have for all inputs and outputs where the marginal element $Q_{0}(a_0|x_0,x_1,\lambda)=Q_{0}(a_0|x_0,\lambda)\neq 0$:
    \begin{equation}
        Q_{01}(a_0,a_1|x_0,x_1,\lambda) = Q_{1|0}(a_1|x_0,x_1,a_0,\lambda)Q_{0}(a_0|x_0,x_1,\lambda) = Q_{1|0}(a_1|x_0,x_1,a_0,\lambda)Q_{0}(a_0|x_0,\lambda),
    \end{equation}
    where we have implicitly defined $Q_{1|0}(a_1|x_0,x_1,a_0,\lambda)$. Note that whenever $Q_{0}(a_0|x_0,\lambda)=0$ for some given inputs and outputs, it also implies $Q_{01}(a_0,a_1|x_0,x_1,\lambda)=0$ and, therefore, the corresponding $\lambda$ index can simply be left out of the summation. Note again that because of $\trace((\Pi_k\otimes\mathbb{1}_Z)\sigma_{W\!Z})\neq 0$ there must be at least one $\lambda$ index such that neither $Q_{01}(a_0,a_1|x_0,x_1,\lambda)$ nor $\tilde{\sigma}_{1}(\lambda,\underline{i})$ is zero. Let us denote these indices by $\lambda_{\underline{k}}$. Thus, for all remaining $\lambda_{\underline{k}}$ indices, we have
    \begin{equation}
        \trace_{W}\!\!\left[(\Pi_{\underline{k}} \!\otimes \! \mathbb{1}_Z) \sigma_{W\!Z}\right]\!=\!  \sum_{\lambda_{\underline{k}}}\!\sum_{a_1,x_1} \! \!\pi_0(\underline{x_0})\pi_1(x_1)r(\lambda_{\underline{k}})  Q_{1|0}(a_1|\underline{x_0},\!x_1,\!\underline{a_0},\!\lambda_{\underline{k}})Q_{0}(\underline{a_0}|\underline{x_0},\lambda_{\underline{k}}) \ketbra{x_1}{x_1}\! \otimes \! \ketbra{a_1}{a_1} \!\otimes \tilde{\sigma}_{1}(\lambda_{\underline{k}},\underline{i}).
    \end{equation}
    Furthermore, 
    \begin{equation}
    \begin{split}
        \trace\left[(\Pi_{\underline{k}} \otimes \mathbb{1}_Z) \sigma_{W\!Z}\right] &=  \sum_{\lambda_{\underline{k}}}\sum_{a_1,x_1} \pi_0(\underline{x_0})\pi_1(x_1)r(\lambda_{\underline{k}}) Q_{1|0}(a_1|\underline{x_0},x_1,\underline{a_0},\lambda_{\underline{k}})Q_{0}(\underline{a_0}|\underline{x_0},\lambda_{\underline{k}}) \trace(\tilde{\sigma}_{1}(\lambda_{\underline{k}},\underline{i}))\\
        &= \sum_{\lambda_{\underline{k}}}\sum_{x_1} \pi_0(\underline{x_0})\pi_1(x_1)r(\lambda_{\underline{k}}) \left(\sum_{a_1}Q_{1|0}(a_1|\underline{x_0},x_1,\underline{a_0},\lambda_{\underline{k}})\right)Q_{0}(\underline{a_0}|\underline{x_0},\lambda_{\underline{k}}) \trace(\tilde{\sigma}_{1}(\lambda_{\underline{k}},\underline{i}))\\
        &= \sum_{\lambda_{\underline{k}}} \pi_0(\underline{x_0})\left(\sum_{x_1}\pi_1(x_1)\right)r(\lambda_{\underline{k}}) Q_{0}(\underline{a_0}|\underline{x_0},\lambda_{\underline{k}}) \trace(\tilde{\sigma}_{1}(\lambda_{\underline{k}},\underline{i}))\\
        &= \sum_{\lambda_{\underline{k}}} \pi_0(\underline{x_0})r(\lambda_{\underline{k}})  Q_0(\underline{a_0}|\underline{x_0},\lambda_{\underline{k}}) \trace(\tilde{\sigma}_{1}(\lambda_{\underline{k}},\underline{i})).
    \end{split}
    \end{equation}
    Thus, 
    \begin{align}
        \sigma_{Z}^{\underline{k}} &= \frac{\trace_{W}((\Pi_{\underline{k}} \otimes \mathbb{1}_Z) \sigma_{W\!Z})}{\trace((\Pi_{\underline{k}} \otimes \mathbb{1}_Z) \sigma_{W\!Z})}= \frac{\trace_{W}((\Pi_{\underline{x_0},\underline{a_0},\underline{i}} \otimes \mathbb{1}_Z) \sigma_{W\!Z})}{\trace((\Pi_{\underline{x_0},\underline{a_0},\underline{i}} \otimes \mathbb{1}_Z) \sigma_{W\!Z})}\\ \nonumber
        &=\frac{\sum_{\lambda_{\underline{k}}}\sum_{a_1,x_1} \pi_0(\underline{x_0})\pi_1(x_1)r(\lambda_{\underline{k}}) Q_{1|0}(a_1|\underline{x_0},x_1,\underline{a_0},\lambda_{\underline{k}})Q_{0}(\underline{a_0}|\underline{x_0},\lambda_{\underline{k}}) \ketbra{x_1}{x_1} \otimes \ketbra{a_1}{a_1} \otimes \tilde{\sigma}_{1}(\lambda_{\underline{k}},\underline{i})}{\sum_{\lambda_{\underline{k}}'} \pi_0(\underline{x_0})r(\lambda_{\underline{k}}')Q_0(\underline{a_0}|\underline{x_0},\lambda_{\underline{k}}') \trace(\tilde{\sigma}_{1}(\lambda_{\underline{k}}',\underline{i}))}\\ \nonumber
        &= \sum_{\lambda_{\underline{k}}}\sum_{a_1,x_1} \frac{r(\lambda_{\underline{k}})Q_{0}(\underline{a_0}|\underline{x_0},\lambda_{\underline{k}}) \trace(\tilde{\sigma}_{1}(\lambda_{\underline{k}},\underline{i}))}{\sum_{\lambda_{\underline{k}}'}r(\lambda_{\underline{k}}')Q_0(\underline{a_0}|\underline{x_0},\lambda_{\underline{k}}') \trace(\tilde{\sigma}_{1}(\lambda_{\underline{k}}',\underline{i}))} \pi_1(x_1)Q_{1|0}(a_1|\underline{x_0},x_1,\underline{a_0},\lambda_{\underline{k}})\ketbra{x_1}{x_1} \otimes \ketbra{a_1}{a_1} \otimes \frac{\tilde{\sigma}_{1}(\lambda_{\underline{k}},\underline{i})}{\trace(\tilde{\sigma}_{1}(\lambda_{\underline{k}},\underline{i}))},
    \end{align}
    where we have simultaneously multiplied and divided by $ \trace(\tilde{\sigma}_{1}(\lambda_{\underline{k}},\underline{i}))\neq0$. Let
    \begin{equation}
        r_{\underline{k}}(\lambda_{\underline{k}})  \coloneqq \frac{r(\lambda_{\underline{k}})Q_{0}(\underline{a_0}|\underline{x_0},\lambda_{\underline{k}}) \trace(\tilde{\sigma}_{1}(\lambda_{\underline{k}},\underline{i}))}{\sum_{\lambda_{\underline{k}}'}r(\lambda_{\underline{k}}')Q_{0}(\underline{a_0}|\underline{x_0},\lambda_{\underline{k}}') \trace(\tilde{\sigma}_{1}(\lambda_{\underline{k}}',\underline{i}))}.
    \end{equation}
    It follows that $r_{\underline{k}}(\lambda_{\underline{k}}) \ge 0$ and that $\sum_{\lambda_{\underline{k}}}r_{\underline{k}}(\lambda_{\underline{k}}) = 1$. Let us also define $Q_{1}(a_1|x_1,\lambda_{\underline{k}},\underline{k}) \coloneqq  Q_{1|0}(a_1|\underline{x_0},x_1,\underline{a_0},\lambda_{\underline{k}})$ and the normalised quantum state
    \begin{equation}
        \sigma_{1}(\lambda_{\underline{k}},\underline{k}) \coloneqq  \frac{\tilde{\sigma}_{1}(\lambda_{\underline{k}},\underline{i})}{\trace(\tilde{\sigma}_{1}(\lambda_{\underline{k}},\underline{i}))}.
    \end{equation}
    Then,
    \begin{equation}
        \sigma_{Z}^{\underline{k}} = \sum_{\lambda_{\underline{k}}}\sum_{a_1,x_1}\pi_1(x_1)r_{\underline{k}}(\lambda_{\underline{k}})Q_{1}(a_1|x_1,\lambda_{\underline{k}},\underline{k}) \ketbra{x_1}{x_1} \otimes \ketbra{a_1}{a_1} \otimes \sigma_{1}(\lambda_{\underline{k}},\underline{k}).
    \end{equation}
    Thus, the $\sigma_{Z}^{k}$ states are CQ-states associated to unsteerable assemblages (in the $A_1|B_1$ split) by the marginal probability distribution $\{\pi_1(x_1)\}_{x_1}$.
\end{proof}

The following lemma proves that the convex combination of CQ-states associated to unsteerable assemblages through the same probability distribution is itself associated to an unsteerable assemblage through the same probability distribution.

\begin{lemma}\label{lemma:convex_sum_CQ_assemblages}
    Let $\{\sigma^{k}\}_{k}$ be a set of CQ-states associated to unsteerable (in the $A|B$ split) assemblages $\{Q^k(a|x)\sigma^k_{a|x}\}_{a,x}$ by the same probability distribution $\{\pi(x)\}_{x}$, and let $\{\alpha_{k}\}_k \subset \mathbb{R}^{+}$ be such that $\sum_k \alpha_k = 1$. Then $\sum_k \alpha_k \sigma^{k}$ is also a CQ-state associated to an unsteerable assemblage by the probability distribution $\{\pi(x)\}_{x}$.
\end{lemma}
\begin{proof}
    According to the hypotheses, for each $k$ there exists an exogenous variable $\lambda_k$, a probability distribution $\{r_k(\lambda_k)\}_{\lambda_k}$, a behaviour $\{Q^{k}(a|x,\lambda_k)\}_{a,x,\lambda_k}$, and quantum states $\{\sigma_B^{k}(\lambda_k)\}_{\lambda_k}\subset\mathcal{D}(\mathcal{H}_B)$ such that:
    \begin{equation}
        \sigma^{k} = \sum_{\lambda_k}\sum_{a,x}\pi(x)r_{k}(\lambda_k)Q^{k}(a|x,\lambda_k) \ketbra{x}{x} \otimes \ketbra{a}{a} \otimes \sigma_{B}^{k}(\lambda_k).
    \end{equation}
    Then,
    \begin{equation}
        \sum_k \alpha_k \sigma^{k} = \sum_{k}\sum_{\lambda_k}\sum_{a,x}\pi(x)\alpha_{k} r_{k}(\lambda_k)Q^{k}(a|x,\lambda_k) \ketbra{x}{x} \otimes \ketbra{a}{a} \otimes \sigma_{B}^{k}(\lambda_k).
    \end{equation}
    
    Let $\lambda' \coloneqq (k,\lambda_k)$, $r(\lambda') = r(k,\lambda_k) \coloneqq \alpha_k r_k(\lambda_k)$, $Q(a|x,\lambda') \coloneqq Q^{k}(a|x,\lambda_k)$, and $\sigma_{B}(\lambda') \coloneqq \sigma_{B}^{k}(\lambda_k)$. We can see that $r(\lambda')\ge 0$ and that $\sum_{\lambda'}r(\lambda') = \sum_k \alpha_k (\sum_{\lambda_k} r_k(\lambda_k)) = \sum_k \alpha_k = 1$. 
    Thus,
    \begin{equation}
        \sum_k \alpha_k \sigma^{k} = \sum_{\lambda'}\sum_{a,x}\pi(x) r(\lambda')Q(a|x,\lambda') \ketbra{x}{x} \otimes \ketbra{a}{a} \otimes \sigma_{B}(\lambda').
    \end{equation}
    Therefore, $\sum_k \alpha_k \sigma^{k}$ is a CQ-state associated to an unsteerable assemblage (in the $A|B$ split) by the probability distribution $\{\pi(x)\}_{x}$.
\end{proof}

The following lemma shows that assemblages remain comparable after using an informationally complete measurement map.

\begin{lemma}\label{lemma:injectivity_ICPOVMs}
Let $\RR=\{P(a|x)\rho_{a|x}\}_{a,x}$ and $\ZZ=\{Q(a|x)\sigma_{a|x}\}_{a,x}$ be assemblages from the $(1,1,m,o,d)$ steering scenario and let $\{\pi(x)\}_x$ be a probability distribution such that $\pi(x)>0$ for all $x$. Let $\{\Gamma_i\}_i$ be an IC-POVM on $\mathcal{H}_B$ and let $\{\Pi_{x,a,i}\}_{x,a,i} \coloneqq \{\ketbra{x}{x} \otimes \ketbra{a}{a} \otimes \Gamma_i\}_{x,a,i}\subset\mathcal{L}(\mathbb{C}^m\otimes\mathbb{C}^o\otimes\mathcal{H}_B)$ be a POVM with $\{\ket{x}\}_x$, $\{\ket{a}\}_a$ being orthonormal bases. Then we have for the associated measurement function $\mathcal{N}$:
\begin{equation}
\mathcal{N}(\rho_{EFB}) = \mathcal{N}(\sigma_{EFB}) \quad \implies \quad \RR= \ZZ,
\end{equation} 
where $\rho_{EFB}$ and $\sigma_{EFB}$ are the associated CQ-states.
\end{lemma}
\begin{proof}
We know that
    \begin{equation}
        \mathcal{N}(\rho_{EFB}) = \sum_{x,a,i} \trace(\Pi_{x,a,i}\rho_{EFB}) \ketbra{x}{x} \otimes \ketbra{a}{a} \otimes \ketbra{i}{i}.
    \end{equation}
    Expanding the trace:
    \begin{equation}
    \begin{split}
        \trace(\Pi_{x,a,i}\rho_{EFB}) &= \sum_{x',a'} \pi(x') P(a'|x') \trace\left( \ketbra{x}{x}\ketbra{x'}{x'} \otimes \ketbra{a}{a}\ketbra{a'}{a'} \otimes \Gamma_i \rho_{a'|x'}\right)\\
        &= \pi(x)P(a|x)\trace(\Gamma_i\rho_{a|x}).
    \end{split}
    \end{equation}
    Thus, 
    \begin{equation}
        \mathcal{N}(\rho_{EFB}) = \sum_{x,a,i} \pi(x)P(a|x)\trace(\Gamma_i\rho_{a|x}) \ketbra{x}{x} \otimes \ketbra{a}{a} \otimes \ketbra{i}{i}.
    \end{equation}
    Analogously,
    \begin{equation}
        \mathcal{N}(\sigma_{EFB}) = \sum_{x,a,i} \pi(x)Q(a|x)\trace(\Gamma_i\sigma_{a|x}) \ketbra{x}{x} \otimes \ketbra{a}{a} \otimes \ketbra{i}{i}.
    \end{equation}
    As $\{\ket{x}\}_x, \{\ket{a}\}_a$ and $\{\ket{i}\}_i$ are orthonormal bases, $\mathcal{N}(\rho_{EFB}) = \mathcal{N}(\sigma_{EFB})$ implies that:
    \begin{equation}\label{eq:equality_traces}
        \pi(x)P(a|x)\trace(\Gamma_i\rho_{a|x}) = \pi(x)Q(a|x)\trace(\Gamma_i\sigma_{a|x}), \qquad \forall x,a,i.
    \end{equation}
    Dividing both sides by $\pi(x)\neq 0$ and summing over $i$:
    \begin{align}\label{eq:equality_p_and_q}
  &P(a|x)\,\tr\Bigl(\sum_i \Gamma_i\rho_{a|x}\Bigr)
    = Q(a|x)\,\tr\Bigl(\sum_i \Gamma_i\sigma_{a|x}\Bigr)\text{,}
    &&\forall x,a, \nonumber\\
  &P(a|x)\,\tr\bigl(\rho_{a|x}\bigr)
    = Q(a|x)\,\tr\bigl(\sigma_{a|x}\bigr)\text{,}
    &&\forall x,a, \\
  &P(a|x)
    = Q(a|x)\text{,}
    &&\forall x,a.\nonumber
\end{align}
    Thus, disregarding the trivial case of $P(a|x)=Q(a|x)=0$ in which the states $\rho_{a|x}$ and $\sigma_{a|x}$ do not exist, we have by eqs.~\eqref{eq:equality_traces} and \eqref{eq:equality_p_and_q} that $\trace(\Gamma_i\rho_{a|x}) =\trace(\Gamma_i\sigma_{a|x})$ for all $x,a,i$. As $\{\Gamma_i\}_i$ is an IC-POVM, this implies that $\rho_{a|x} = \sigma_{a|x}$ for all $x,a$. Thus,
    \begin{equation}
        P(a|x)\rho_{a|x} = Q(a|x)\sigma_{a|x}, \qquad \forall x,a.
    \end{equation}
    Therefore, $\RR= \ZZ$.    
\end{proof}

The following lemma guarantees the finiteness of the relative entropy of steering, and thus, for our purposes, the finiteness of the terms in Piani's inequality.

\begin{lemma}\label{lemma:finite_assemblage}
    Let $\RR_{01}=\{P_{01}(a_0,a_1|x_0,x_1)\rho_{a_0,a_1|x_0,x_1}\}_{a_0,a_1,x_0,x_1}$ be any assemblage from the $(2,2,m,o,d)$ steering scenario. Then the relative entropy of steering of $\RR_{01}$ is finite.
\end{lemma}

\begin{proof}
    By definition, we have:
    \begin{equation}
        E_{\mathrm{UR}}\left(\RR_{01}\right)\coloneqq \inf_{\ZZ_{01} \in \mathrm{UR_{ns}}} S_{\mathrm{a}}(\RR_{01}||\ZZ_{01})=\inf_{\ZZ_{01} \in \mathrm{UR_{ns}}}\sup_{\pi_{01}} S_{\mathrm{q}}(\rho_{W\!Z}||\sigma_{W\!Z})\text{,}
    \end{equation}
    where $\pi_{01}$ is any joint probability distribution on the choices of the two Alices and $\rho_{W\!Z},\sigma_{W\!Z}$ are the related CQ-states.

    In general, the relative entropy of assemblages is always non-negative but can be infinite. However, it is easy to see that it must be finite, as one can always choose a special uniform element $\hat{\ZZ}_{01}=\{\hat{Q}_{01}(a_0,a_1|x_0,x_1)\hat{\sigma}_{a_0,a_1|x_0,x_1}\}_{a_0,a_1,x_0,x_1}$ from $\mathrm{UR_{ns}}$ such that for all inputs and outputs:
    \begin{equation}
        \hat{Q}_{01}(a_0,a_1|x_0,x_1){\hat{\sigma}_{a_0,a_1|x_0,x_1}}\coloneqq\frac{1}{o^2}\frac{\mathbb{1}_{B_0}\otimes \mathbb{1}_{B_1}}{d^2}\text{,}
    \end{equation}
    where $o$ is the number of outputs of any Alice and $(\mathbb{1}_{B_0}\otimes \mathbb{1}_{B_1})/d^2$ is simply the maximally mixed state on the Hilbert spaces of the Bobs. It is easy to check that this is an $\mathrm{UR_{ns}}$ assemblage as it can be described by a singleton exogenous variable $\lambda$ and the behaviour itself factorises into the product of individual boxes for the parties $A_0$ and $A_1$, all with uniform outputs regardless of inputs. This leads to:
    \begin{align}
    \rho_{W\!Z} &\coloneqq \underset{\substack{a_0,a_1\\x_0,x_1}}{\sum} \pi_{01}(x_0,x_1)P_{01}(a_0,a_1|x_0,x_1) \ketbra{x_0,x_1}{x_0,x_1} \otimes \ketbra{a_0,a_1}{a_0,a_1} \otimes \rho_{a_0,a_1|x_0,x_1},\\
    \hat{\sigma}_{W\!Z} &\coloneqq \underset{\substack{x_0,x_1}}{\sum} \pi_{01}(x_0,x_1)\frac{1}{o^2}\ketbra{x_0,x_1}{x_0,x_1} \otimes \mathbb{1}_{F_0}\otimes\mathbb{1}_{F_1} \otimes \frac{\mathbb{1}_{B_0}\otimes\mathbb{1}_{B_1}}{d^2},
    \end{align}
    where, in the second equation, we have already gone through the summation for the $a_0,a_1$ indices which are related to the Hilbert spaces $\mathcal{H}_{F_0}$ and $\mathcal{H}_{F_1}$.
    
    We can show the following for any $\rho_{W\!Z},\sigma_{W\!Z}$ comparable CQ-states:
    \begin{align}
            &S_q(\rho_{W\!Z}||\sigma_{W\!Z})=\sum_{x_0,x_1}\pi_{01}(x_0,x_1)S_{\mathrm{q}}(\rho_{W'\!Z'}(x_0,x_1)||\sigma_{W'\!Z'}(x_0,x_1))=\\
            &=\sum_{x_0,x_1}\pi_{01}(x_0,x_1)\left(S(P_{01}(.,.|x_0,x_1)||Q_{01}(.,.|x_0,x_1))+\sum_{a_0,a_1}P_{01}(a_0,a_1|x_0,x_1)S_{\mathrm{q}}(\rho_{a_0,a_1|x_0,x_1}||\sigma_{a_0,a_1|x_0,x_1})\right)\text{,}\nonumber
    \end{align}
    where $\rho_{W'\!Z'}(x_0,x_1)$ and $\sigma_{W'\!Z'}(x_0,x_1)$ are the $x_0,x_1$ input dependent CQ-states just like in Equation \ref{eq:inputdependentcqstates}, which warrants the change from the $W=(E_0,F_0,B_0)$ index to the new $W'\coloneqq(F_0,B_0)$ index, and similarly for $Z$ and $Z'$.

    In our case, with $\hat{\ZZ}_{01}$ we have the following results for the terms on the right-hand side:
    \begin{align}
        S(P_{01}(.,.|x_0,x_1)||\hat{Q}_{01}(.,.|x_0,x_1))&=\log(o^2)-H(P(.,.|x_0,x_1))\leq \log(o^2)\text{,}\\
        S_{\mathrm{q}}(\rho_{a_0,a_1|x_0,x_1}||\hat{\sigma}_{a_0,a_1|x_0,x_1})&=\log(d^2)-H_{q}(\rho_{a_0,a_1|x_0,x_1})\leq\log(d^2)\text{,}
    \end{align}
    where $H$ and $H_q$ are the Shannon entropy and the von Neumann entropy, respectively. We have also noted that for our scenario, the Shannon entropy takes values between $0$ and $\log(o^2)$, while the von Neumann entropy takes values between $0$ and $\log(d^2)$. This leads to:
    \begin{equation}
        S_q(\rho_{W\!Z}||\hat{\sigma}_{W\!Z})\leq \sum_{x_0,x_1}\pi_{01}(x_0,x_1)\left(\log(o^2)+\sum_{a_0,a_1}P_{01}(a_0,a_1|x_0,x_1)\log(d^2)\right)\leq \log(o^2d^2)\text{.}
    \end{equation}
    The above result is true regardless of the probability distribution $\pi_{01}$ and we have shown that $\hat{\ZZ}_{01}\in\mathrm{UR_{ns}}$. Thus, we have:
    \begin{equation}
        E_{\mathrm{UR}}\left(\RR_{01}\right)\coloneqq \inf_{\ZZ_{01} \in \mathrm{UR_{ns}}} S_{\mathrm{a}}(\RR_{01}||\ZZ_{01})\leq \log(o^2d^2)\text{,}
    \end{equation}
    proving that it must be finite.
\end{proof}

The following lemma is an elementary result from real analysis about the exchange of the infimum and supremum. See, e.g.~Lemma 36.1.~in \cite{rockafellar1973convex} or Corollary 3.8.4.~in \cite{riehl2016category}. We quote it here directly:

\begin{lemma}\label{lemma:analysis}
    Let $X$ and $Y$ be nonempty sets, and let $f:X\times Y\to \mathbb{R}$ be any function. Then we have that:
        \begin{equation}
         \inf_{y\in Y}\left(\sup_{x\in X} f(x,y)\right)\geq \sup_{x\in X}\left(\inf_{y\in Y} f(x,y)\right)\text{,}
        \end{equation}
    whenever these infima and suprema exist.
\end{lemma}

\begin{proof}
    We omit the proof and refer to the above-cited sources. We note that in our use case, we have that the range of the quantum KL-divergence is the extended, non-negative real line segment $[0,\infty]$. Thus, we have that every infimum and supremum exists with the understanding that they might be defined as $+\infty$.
\end{proof}

We are now able to show the main result of this section:
\propDilationEntropyAssemblages*
\begin{proof}
Let $\RR'=\RR'_{01}\eqqcolon\{P'_{01}(a_0,a_1|x_0,x_1)\rho'_{a_0,a_1|x_0,x_1}\}_{a_0,a_1,x_0,x_1}$ be the broadcast version of the assemblage $\RR\eqqcolon\{P(a|x)\rho_{a|x}\}_{a,x}$ which is steerable in the $A|B$ split. Note that $\RR'_{01}$ is non-signalling in the $0|1$ split. By definition,
\begin{equation}
\begin{split}
E_{\mathrm{UR}}(\RR'_{01}) &\coloneqq \inf_{\ZZ_{01} \in \mathrm{UR_{ns}}} \sup_{\pi_{01}} S_{\mathrm{q}}(\rho_{W\!Z}'||\sigma_{W\!Z})\text{,}
\end{split}
\end{equation}
where $\ZZ_{01}\coloneqq\{Q_{01}(a_0,a_1|x_0,x_1) \sigma_{a_0,a_1|x_0,x_1}\}_{a_0,a_1,x_0,x_1}$ is from $\mathrm{UR_{ns}}$ and $\rho_{W\!Z}'$ and $\sigma_{W\!Z}$ are the associated CQ-states.

We have the following identities:
\begin{align}
    E_{\mathrm{UR}}(\RR'_{01}) &\coloneqq \inf_{\ZZ_{01} \in \mathrm{UR_{ns}}} \sup_{\pi_{01}} S_{\mathrm{q}}(\rho_{W\!Z}'||\sigma_{W\!Z})\\
    &\geq \inf_{\ZZ_{01} \in \mathrm{UR_{ns}}} \sup_{\pi_{0}\pi_{1}} S_{\mathrm{q}}(\rho_{W\!Z}'||\sigma_{W\!Z})\\
    &\geq \inf_{\ZZ_{01} \in \mathrm{UR_{ns}}} \sup_{\pi_0\pi_1} \left[S_{\mathrm{q}}(\mathcal{N}(\rho_W')||\mathcal{N}(\sigma_W)) + S_{\mathrm{q}}\left(\rho_Z'\middle|\middle| \sum_{k} \alpha_{k} \sigma_{Z}^{k}\right) \right]\\
    &=\inf_{\ZZ_{01} \in \mathrm{UR_{ns}}} \sup_{\pi_0} \left[S_{\mathrm{q}}(\mathcal{N}(\rho_W')||\mathcal{N}(\sigma_W)) + \sup_{\pi_1}S_{\mathrm{q}}\left(\rho_Z'\middle|\middle| \sum_{k} \alpha_{k} \sigma_{Z}^{k}\right) \right]\text{.}
\end{align}

In the second line, we are taking the supremum over the smaller set of product probability distributions, i.e.~when $\{\pi_{01}(x_0,x_1)\}_{x_0,x_1} = \{\pi_0(x_0)\pi_1(x_1)\}_{x_0,x_1}$. In the third line, we are using the Piani inequality (eq.~\eqref{eq:Piani's_Result}), with a measurement map $\mathcal{N}:\mathcal{L}(\mathcal{H}_W)\to\mathcal{L}(\mathbb{C}^{m} \otimes \mathbb{C}^{o} \otimes  \mathbb{C}^{n})$ associated to the POVM $\{\Pi_k\}_{k}\coloneqq\{\ketbra{x_0}{x_0} \otimes \ketbra{a_0}{a_0} \otimes \Gamma_i\}_{x_0,a_0,i}\subset \mathcal{L}(\mathbb{C}^{m} \otimes \mathbb{C}^{o} \otimes \mathcal{H}_{B_0})$, where $\{\ket{x_0}\}_{x_0}$ and $\{\ket{a_0}\}_{a_0}$ are orthonormal bases and $\{\Gamma_i\}_{i=1}^{n}$ is an IC-POVM. Finally, in the fourth line, we use the fact that $\rho_W'$ and $\sigma_W'$ do not depend on the probability distribution $\pi_1$.

We continue by invoking \textbf{Lemma \ref{lemma:analysis}} to exchange the outer infimum and supremum:
\begin{align}
    E_{\mathrm{UR}}(\RR'_{01}) &\geq \inf_{\ZZ_{01} \in \mathrm{UR_{ns}}} \sup_{\pi_0} \left[S_{\mathrm{q}}(\mathcal{N}(\rho_W')||\mathcal{N}(\sigma_W)) + \sup_{\pi_1}S_{\mathrm{q}}\left(\rho_Z'\middle|\middle| \sum_{k} \alpha_{k} \sigma_{Z}^{k}\right) \right]\\
    &\geq \sup_{\pi_0} \inf_{\ZZ_{01} \in \mathrm{UR_{ns}}} \left[S_{\mathrm{q}}(\mathcal{N}(\rho_W')||\mathcal{N}(\sigma_W)) + \sup_{\pi_1}S_{\mathrm{q}}\left(\rho_Z'\middle|\middle| \sum_{k} \alpha_{k} \sigma_{Z}^{k}\right) \right]\\
    &\geq \sup_{\pi_0} \left[\inf_{\ZZ_{01} \in \mathrm{UR_{ns}}} S_{\mathrm{q}}(\mathcal{N}(\rho_W')||\mathcal{N}(\sigma_W)) + \inf_{\ZZ_{01} \in \mathrm{UR_{ns}}} \sup_{\pi_1}S_{\mathrm{q}}\left(\rho_Z'\middle|\middle| \sum_{k} \alpha_{k} \sigma_{Z}^{k}\right) \right]\text{.}
\end{align}

In the third line, we have used the fact that the infimum of a sum is larger than or equal to the sum of the infima of the parts. From \textbf{Lemma \ref{lemma:finite_assemblage}} we have that, in our case, Piani's inequality is not vacuous, leading to $\tr((\Pi_k\otimes\mathbb{1}_Z)\sigma_{W\!Z})=0$ only if $\alpha_k=0$. Thus we can invoke \textbf{Lemmas \ref{lemma:post_measurement_assemblages_unsteerable}} and \textbf{\ref{lemma:convex_sum_CQ_assemblages}} and \textbf{Definition \ref{def:urns}} to substitute general unsteerable (US) assemblages $\tilde{\ZZ}_{0}$ and $\tilde{\ZZ}_1$, in place of the marginals of the $\mathrm{UR_{ns}}$ assemblage $\ZZ_{01}$:
\begin{align}
    E_{\mathrm{UR}}(\RR'_{01}) &\geq \sup_{\pi_0} \left[\inf_{\ZZ_{01} \in \mathrm{UR_{ns}}} S_{\mathrm{q}}(\mathcal{N}(\rho_W')||\mathcal{N}(\sigma_W)) + \inf_{\ZZ_{01} \in \mathrm{UR_{ns}}} \sup_{\pi_1}S_{\mathrm{q}}\left(\rho_Z'\middle|\middle| \sum_{k} \alpha_{k} \sigma_{Z}^{k}\right) \right]\\
    &\geq \sup_{\pi_0} \left[\inf_{\tilde{\ZZ}_{0} \in \mathrm{US}} S_{\mathrm{q}}(\mathcal{N}(\rho_W')||\mathcal{N}(\tilde{\sigma}_W)) + \inf_{\tilde{\ZZ}_{1} \in \mathrm{US}} \sup_{\pi_1}S_{\mathrm{q}}\left(\rho_Z'\middle|\middle| \tilde{\sigma}_Z\right) \right]\\
    &=\sup_{\pi_0} \inf_{\tilde{\ZZ}_{0} \in \mathrm{US}} S_{\mathrm{q}}(\mathcal{N}(\rho_W')||\mathcal{N}(\tilde{\sigma}_W)) + \inf_{\tilde{\ZZ}_{1} \in \mathrm{US}} \sup_{\pi_1}S_{\mathrm{q}}\left(\rho_Z'\middle|\middle| \tilde{\sigma}_Z\right)\text{,}
\end{align}
where $\tilde{\sigma}_W,\tilde{\sigma}_W$ are the CQ-states related to the general unsteerable assemblages $\tilde{\ZZ}_{0}$ and $\tilde{\ZZ}_1$. In the last line, we have used the fact that the right summand is no longer dependent on $\pi_0$.

Let us now focus on the left summand:
\begin{align}
\sup_{\pi_0} \inf_{\tilde{\ZZ}_{0} \in \mathrm{US}} S_{\mathrm{q}}(\mathcal{N}(\rho_W')||\mathcal{N}(\tilde{\sigma}_W)) \geq \sup_{\pi_0>0} \inf_{\tilde{\ZZ}_{0} \in \mathrm{US}} S_{\mathrm{q}}(\mathcal{N}(\rho_W')||\mathcal{N}(\tilde{\sigma}_W))= \sup_{\pi_0>0} S_{\mathrm{q}}(\mathcal{N}(\rho_W')||\mathcal{N}(\overline{\sigma}_W))>0\text{.}
\end{align}
The first inequality is true since the supremum is taken over the smaller set of strictly positive probability distributions. In the second equality we use the fact that the infimum is reached at a CQ-state $\overline{\sigma}_W$ related to an unsteerable assemblage $\overline{\ZZ}_1$. Similarly to the case of behaviours explored in \textbf{Lemma \ref{lemma:MinEntropy}} this is also true given that unsteerable assemblages form a compact set and the quantum KL-divergence is lower semi-continuous. Finally, given a fixed unsteerable assemblage $\overline{\sigma}_W$, a strictly positive probability distribution $\pi_0$ and an informationally complete measurement map $\mathcal{N}$, \textbf{Lemma \ref{lemma:injectivity_ICPOVMs}} can be invoked to show that since $\rho'_W$ is steerable and $\overline{\sigma}_W$ is unsteerable we have that $\rho'_W\neq\overline{\sigma}_W$ leading to $\mathcal{N}(\rho'_W)\neq\mathcal{N}(\overline{\sigma}_W)$ leading to their KL-divergence not being zero.

Finally, keeping in mind that \textbf{Lemma \ref{lemma:finite_assemblage}} guarantees that, in our case, Piani's inequality is not vacuous, we have that
\begin{align}
    E_{\mathrm{UR}}(\RR'_{01}) > \inf_{\tilde{\ZZ}_{1} \in \mathrm{US}} \sup_{\pi_1}S_{\mathrm{q}}\left(\rho_Z'\middle|\middle| \tilde{\sigma}_Z\right)\eqqcolon E_{\mathrm{UR}}(\RR'_1)= E_{\mathrm{UR}}(\RR)\text{,}
\end{align}
where we have used that \textbf{Lemma \ref{lemma:Reduced_CQ-states}} applies given that $\RR'_{01}$ is non-signalling in the $0|1$ split. Finally, we have that if $\RR'_{01}$ is a broadcast version of $\RR$, then $\RR'_1 = \RR$.
\end{proof}

\end{widetext}
\end{document}